\documentclass[%
aps, a4paper, twocolumn,11pt
amsmath,amssymb
aps,floatfix,superscriptaddress,
]
{revtex4-1}
\pdfoutput = 1
\usepackage{amsmath}
\usepackage[toc,page]{appendix}
\usepackage[utf8]{inputenc}
\usepackage[english]{babel}
\usepackage{amsfonts}
\usepackage{qcircuit}
\usepackage[T1]{fontenc}
\usepackage{xkeyval}
\usepackage{tikz}
\usepackage{placeins}
\usetikzlibrary{arrows.meta, positioning, shapes.geometric}
\usepackage{graphicx}
\usepackage{pgfplots}
\pgfplotsset{compat=1.18}
\usepackage[export]{adjustbox}
\usepackage{soul}
\usepackage{dcolumn}
\usepackage{bm}
\usetikzlibrary{calc}
\usepackage{caption}
\usepackage{subcaption}
\usepackage[margin=1in]{geometry}
\usepackage{amssymb}
\usepackage{overpic}
\usepackage{epstopdf}
 \usepackage[usenames,dvipsnames]{pstricks}
 \usepackage{epsfig}
 \usepackage{pst-grad} 
 \usepackage{pst-plot} 
 \usepackage[space]{grffile} 
 \usepackage{etoolbox} 
 \makeatletter 
 \patchcmd\Gread@eps{\@inputcheck#1 }{\@inputcheck"#1"\relax}{}{}
 \makeatother
\usepackage{booktabs}
\usepackage{boxhandler}
\usepackage{amsmath}
\usepackage{braket}
\usepackage{amsmath}
\usepackage{qcircuit}
\usepackage{float}
\makeatletter
\let\newfloat\newfloat@ltx
\makeatother
\usepackage{algorithm}
\usepackage{algpseudocode}
\usepgfplotslibrary{groupplots}

\usepackage{hyperref}
\hypersetup{
  colorlinks   = true,    
  urlcolor     = blue,   
  linkcolor    = blue,   
  citecolor    = blue     }

\usepackage{multirow}
\usepackage{enumitem}
\usepackage[normalem]{ulem}

\usepackage[amsthm]{ntheorem}

\newcommand{\F}{\mathbb{F}}

\theoremstyle{plain}
\newtheorem{theorem}{Theorem}[section]
\newtheorem{lemma}[theorem]{Lemma}
\newtheorem{corollary}[theorem]{Corollary}
\newtheorem{proposition}[theorem]{Proposition}

\theoremstyle{definition}

\newtheorem{remark}[theorem]{Remark}

\newtheorem{definition}[theorem]{Definition}

\numberwithin{equation}{section}
\newtheorem{example}[theorem]{Example}

\newcommand{\Span}{{\rm Span }}
\newcommand{\wt}{{\rm wt}}

\newcommand{\Z}{\mathbb{Z}}

\newcommand{\CNOT}{{\rm CNOT}}

\begin{document}

\title{Quantum Codes with Arbitrary 
Z-Rotation logical Gates and Applications to Fault-Tolerant Code Switching}

\author{Reza {Dastbasteh}}
\email{rdastbasteh@unav.es}
\affiliation{Department of Basic Sciences, Tecnun - University of Navarra, San Sebastian, Spain.}
\author{Ruben M. Otxoa}
%\email{ro274@cam.ac.uk}
\affiliation{Hitachi Cambridge Laboratory, J. J. Thomson Avenue, Cambridge, United Kingdom.}

\author{Pedro M. Crespo}
%\email{pcrespo@unav.es}
\affiliation{Department of Basic Sciences, Tecnun - University of Navarra, San Sebastian, Spain.}
\author{Josu {Etxezarreta Martinez}}
\email{jetxezarreta@unav.es}
\affiliation{Department of Basic Sciences, Tecnun - University of Navarra, San Sebastian, Spain.}

%%%%%%%%%%%%%%%%%%%%%%%%%%%%
\begin{abstract}
%%%%%%%%%%%%%%%%%%%%%%%%%%%%%
A technique for realizing a universal set of fault-tolerant quantum operations is the code switching method, which leverages two quantum codes with complementary sets of transversal gates.
To date, the application of this technique has been largely limited to families of color codes supporting a logical $T$ gate.
No analogous code switching protocols exist for many other prominent families, such as rotated surface codes, or for finer $Z$-rotation gates.
In this work, we first utilize the doubling technique as a unified framework to construct a class of quantum color codes encoding a single logical qubit with an arbitrarily large minimum distance, enabling the transversal realization of arbitrary small logical $Z$-rotation gates.
We investigate the structural properties of this code family, demonstrating that they improve upon the parameters of state-of-the-art triorthogonal codes, achieve lower qubit overhead compared to certain known color codes, and admit single-shot decoding of $Z$-syndromes via meta-checks.
Furthermore, we show that this framework extends beyond color codes; specifically, it enables the generation of $r$-orthogonal quantum codes, $r \ge 2$, that inherit the local geometry of rotated surface codes.
We then provide an overhead optimization protocol alongside several candidate codes tailored for realizing logical $Z$-rotation gates within rotated surface codes.
Finally, we extend the fault-tolerant code switching protocol based on transversal CNOT gates to incorporate fault-tolerant realization of $Z$-rotation gates at any level of the Clifford hierarchy for geometries compatible with rotated surface codes.
We present the first demonstration of fault-tolerant magic state preparation by means of code switching within a distance-three rotated surface code using a total footprint of only 45 physical qubits, and evaluate its performance through a simulation.
\end{abstract}

\maketitle

%\tableofcontents

\section{Introduction}
Quantum error correcting codes are the key ingredient  for enabling large scale quantum computation and exploiting the potential advantages of quantum computers by protecting quantum information via lowering logical error rates. 
This is done by encoding the code space into a larger Hilbert space as a trade off for achieving error detection and correction \cite{Gottesman,Calderbank}.

Besides the mentioned quantum error suppression feature, logical operators must be enabled over quantum codes as otherwise they can only be used as quantum memories. 
Such quantum operations must be implemented in a fault-tolerant (FT) way that is consistent with the error protection scheme of the underlying code. 
Among different FT protocols, transversal gates are of main interest as they are inherently FT, can be parallelized, and are consistent with the Steane or Knill's type quantum error correction algorithms \cite{steane1997active,knill2005scalable,heussen2025efficient}. 
However,  the Eastin-Knill no-go theorem limits the realization of a universal set of gates in a single error correcting code transversally by stating that no quantum code can have a universal set of transversal gates \cite{eastin2009restrictions}. 
A generalization of the Eastin-Knill no-go theorem  for other types of FT gates, such as fold-transversal and automorphism-based gates, is discussed in \cite{chakraborty2026no}. 

One way to overcome this limitation is by sticking to the codes that realize logical Clifford gates transversally, and attaining a non-Clifford gate via an alternative FT approach. 
This is mainly because logical Clifford operations have a relatively low cost as they
can be implemented 
through many stabilizer codes (surface codes, color codes, LDPC codes, etc)
either
transversally or fold transversally, by simple code deformation, or by very cheap constant-depth circuits \cite{bombin2006topological,breuckmann2024fold,horsman2012surface,moussa2016transversal,zheng2020constant}. 
On the other hand, logical non-Clifford
gates with transversal implementations only exist in quantum codes that have a rare symmetry. Moreover, such gates usually have a relatively high cost that may exceed the one of Clifford operators by orders of magnitude. 
The most common technique for realizing non-Clifford gates is based on magic state distillation protocols, code switching,  and magic state cultivation. 

In magic state distillation \cite{bravyi2005universal,knill2004fault,litinski2019magic,beverland2021cost,bravyi2012magic,jones2013multilevel,jacinto2026exploringlandscapecompactmagicstate}, noisy magic states are injected into a quantum code and purified through distillation protocols. 
These protocols consume multiple auxiliary noisy magic states to produce a smaller number of magic states with significantly lower infidelity.
One main challenge in magic state distillation is its considerable space-time overhead, and the fact that injection of a magic state is not fault tolerant. 
Thus, reducing the error rate of such a non-Clifford resource, to a level comparable to the error rate of other Clifford logical operations, can lead to a considerable amount of extra complexity and overhead.

Magic state cultivation (MSC) \cite{gidney2024magic,chen2025efficient,chen2026efficientmagicstatecultivation,claes2025cultivating,vaknin2025efficient} has been proposed as an alternative approach in which the fidelity of a magic state is progressively improved through a sequence of fault-tolerant Clifford operations and measurements. 
Rather than relying on large block distillation protocols, MSC gradually refines an initially noisy resource state while maintaining compatibility with the underlying quantum code. 
Although MSC drastically reduces the space-time overhead of factory-style distillation, its resource costs remain significant, and it introduces new architectural challenges like dynamic code growth, patch grafting, and FT decoding.

Another alternative technique is the code switching, also known as dimension jump (or its special case namely gauge-fixing) \cite{anderson2014fault,bombin2015gauge,bombin2016dimensional,heussen2025efficient,daguerre2025code,jiao2025low,campbell2017roads,golowich2025constant,ouyang2025measurement}, in which two different quantum error-correcting codes are employed to realize a FT universal set of gates. 
Typically, one of the codes efficiently supports the implementation of a generator set for Clifford gates, while the other admits the transversal or other FT implementation of a non-Clifford gate. 
In this protocol, the logical information is fault-tolerantly transferred between the two codes in order to apply the non-Clifford operation. 
Consequently, non-Clifford gates can be realized without the need to prepare high-fidelity magic states, potentially reducing the overhead associated with magic state distillation.

Almost all currently-known code switching protocols are based on the color code family  \cite{Color2,eczoo_color}, where the logical state is transferred between 2D and 3D color codes \cite{anderson2014fault, bombin2015gauge,bombin2016dimensional,heussen2025efficient,daguerre2025code,butt2024fault}. 
The former codes support transversal implementations of Clifford gates, while the latter enable transversal logical $T$ gates.
Recently a similar protocol for the realization of a logical CCZ through a one-way transversal CNOT gate between 3D and 2D lifted product and hypergraph product codes has extended the idea behind the family of color codes \cite{li2025transversal,tan2025single, ouyang2025measurement}. 
The primary focus of this work is also on code switching protocols. 
Specifically, it aims to systematically generate suitable codes and expand the families of quantum codes available for this purpose. 

In general, the resource estimation comparison between the three mentioned methods in a fair setting seems impossible as they are distinct 
in the ingredient codes, noise model, decoder and post-selection, and other underlying  assumptions. 
However, almost all the studies agree on the very low space-time overhead of code switching protocols \cite{li2025transversal,heussen2025efficient,daguerre2025code,gao2026resource}. 
More particularly, it is reported in \cite{gao2026resource} that
among the three proposals, and considering the available data in the literature:
\begin{itemize}
    \item magic state distillation stands out by reaching the lowest overall logical output error regime, 
    \item code switching protocols achieve the lowest single-attempt cost and the smallest qubit footprint, and
    \item  MSC bridges these approaches by contributing within intermediate error ranges and space-time overhead.
\end{itemize}

Magic state distillation and code switching both rely on the existence of quantum codes with supporting a logical $T$ gate that can be implemented transversally or in other FT manners.  
The characterization of such codes has been discussed in \cite[Theorem 1]{rengaswamy2020optimality}. 
In particular, the most general family that enables transversal logical $T$ gate through the transversal physical $T$ gate is the so-called triorthogonal code family \cite{bravyi2012magic}. 
Motivated by this and the triorthogonal characterization of quantum codes given in \cite{bravyi2012magic}, new families and examples of quantum triorthogonal codes have been designed \cite{bravyi2015doubled, rengaswamy2020optimality,jain2025transversal,berardini2025asymptotically,nezami2022classification,camps2024algebraic,camps2026transversal}. 
Among such constructions, the quantum doubling construction \cite{bravyi2015doubled,jain2025transversal} paved the path for a recursive construction of triorthogonal codes with the aid of self-dual binary linear codes. 
Also, existence of asymptotically good  quantum (LDPC) codes with non vanishing rate and non vanishing minimum distance (or non vanishing rate and growing minimum distance) supporting a non-Clifford logical gate have been developed very recently \cite{Asym1,Asym2,Asym3, Asym4,Asym5}.
Furthermore, a bi-directional switching protocol via quantum state
teleportation circuits at the logical level for the triorthogonal codes constructed from the doubling construction is discussed in \cite{heussen2025efficient}.
Recently, experimental implementations of the code switching protocol have been demonstrated to prepare a high-fidelity encoded logical magic state. \cite{daguerre2025experimental,pogorelov2025experimental}. 

Currently, the most commonly considered gate set for realizing universal quantum computation on quantum codes is Clifford+T. 
While Clifford+T is a universal gate set, the practical efficiency of implementing complex operations depends heavily on the FT decomposition of such operations into Clifford+T, and also the number of qubits involved.
For example, multi-controlled Toffoli gates cannot be implemented directly using only the native Clifford+T operations on the same set of qubits. 
Instead, a FT decomposition requires additional ancilla qubits to break these gates into sequences of Clifford and $T$ operations, which increases both the number of logical qubits and the circuit depth, and consequently reduces the achievable fidelity of the overall FT implementation.

An alternative proposal is to complement the $T$ gates with $R_Z(\theta)$, namely smaller rotations around the $Z$-axis of the Bloch sphere by an angle $\theta$,  
forming a Clifford+$R_Z(\theta)$ gate set.
This extension allows more efficient decompositions of general gates, such as multi-controlled Toffoli gates. 
In this framework, the decomposition requires fewer elementary gates, a smaller circuit (and non-Clifford) depth and no additional ancilla qubits \cite{biswal2019techniques, forest2015exact}.
More interestingly, recently, it has been suggested that small $R_Z(\theta)$ can be substantially cheaper than previously assumed indicating that many FT quantum algorithms, those with large numbers of small $Z$-rotations, may require significantly fewer non-Clifford resources than standard estimates predict \cite{bothe2026more}.
This promotes the need for quantum codes with small angle $Z$-rotation logical gates, and novel protocols for preparation of such non-Clifford resources with high fidelity. 
This is another main motivation of the code families and examples presented in this paper. 
In particular, for {\em each desired rotation angle in the form $\theta=\frac{\pi}{2^r}$ and minimum distance $d$}, respectively, we give a quantum code with a logical $R_Z(\theta)$ gate and minimum distance $d$ that can be implemented transversally, and be used for preparation  of high fidelity magic states.

%%%%%%%%%%%%%%%%%%%%%%%%%%%%%%%%%%%%%%%
\subsection{Our main contributions}
%%%%%%%%%%%%%%%%%%%%%%%%%%%%%%%%%%%%%%%%
In spite of the previously mentioned advances in the code switching protocol via transversal CNOT gate for realizing logical $Z$-rotation gates  \cite{heussen2025efficient,daguerre2025code}, its applications are remained limited to (a) the color code family and (b) solely for realizing the $T$ gate. 
Motivated by these limitations, we reveal new applications of the doubling method in constructing divisible and orthogonal codes which could serve as ingredients for the code switching protocol, and also for realizing arbitrary fine logical $Z$-rotation gates transversally. 
We highlight a summary of our main results below.
\begin{itemize}
\item We construct quantum codes that admit transversal logical $Z$-rotation gates at arbitrary levels of the Clifford hierarchy by employing a modification of the doubling construction.
In particular, we recursively construct a family of $2^r$-divisible quantum (color) codes with parameters 
\[
[\![S_r(k),1,2k-1]\!]
\]
that realizes logical $R_Z( \frac{\pi}{2^{r-1}})$ through the transversal action of them to the data qubits, where
\[S_r(k)=\sum_{i=0}^{r} 2^i \binom{i+k-2}{i}.
\] 
As an application, we present parameters of small instances of codes inside such family and show examples that improve the parameters of previously-known triorthogonal codes in the literature. 
We also show that this family requires less qubit overhead compared to that of some other well-known code families (see Tables \ref{tab: comparison} and \ref{tab:parameters small}).
Finally, we show that each quantum code constructed from the doubling construction satisfies the possession of $Z$-type meta-checks making the above quantum families viable for single-shot decoding of $Z$-syndromes. 
\item We show that our approach is not restricted to the family of color codes. 
In particular, we adjust our method to construct a family of quantum codes with parameters 
\[
[\![S_r(k),1,2k-1]\!]
\]
that are $r$-orthogonal and preserve the local geometry of rotated surface codes, where
\[
\begin{split}
S_r(k) &= 2^{r+2} \binom{k+r-1}{r+1} \\&+ \sum_{i=0}^{r-1} 2^i \binom{k+i-2}{i}.
\end{split}
\]
As an application, we construct quantum self-orthogonal and triorthogonal codes with the local geometry of rotated surface codes (see Table \ref{tab:surface}) and provide several specific examples.

\item We propose methods for qubit overhead reduction of the above quantum family. 
This is based on carefully concatenating rotated surface codes and punctured rotated surface codes together. 
This leads to a more optimized constructions of quantum self-orthogonal and triorthogonal codes with local geometry of rotated surface codes. 
This includes triorthogonal codes with parameters $[\![31,1,3]\!]$ and $[\![113,1,5]\!]$.

\item We extend the application of the code-switching protocol to quantum codes that have a local geometry of another desired code. 
In particular, one can obtain code-switching protocols for FT implementation of $S$, $T$, or any other $Z$-rotation gate in color codes, rotated surface codes, or extend the idea to other quantum codes of interest.
\item We demonstrate that the proposed code-switching protocol can be realized fault tolerantly for distance $3$ with a remarkably small hardware footprint, requiring only 45 physical qubits in a complete circuit-level simulation.
In particular, we successfully prepare the $S\ket{+}$ state in a $d=3$ rotated surface code as a proxy of the $T\ket{+}$ magic state.  
\end{itemize}

{\em Application.} Immediate applications of the above results are in
\begin{itemize}
    \item presenting many quantum codes with arbitrary minimum distances, local geometries, and $Z$-rotation gates for magic state distillation, cultivation, and code switching. 
    \item FT realization of Clifford+$R_Z(\theta)$ gate sets that have the potential of significantly reducing the cost of non-Clifford resources for universal quantum computation.  
\end{itemize}

This paper is organized as follows. A brief overview of quantum codes and the doubling construction of triorthogonal codes is provided in Section \ref{sec: preliminaries}. In Section \ref{Sec: general family}, we extend the application of the doubling construction to build quantum codes with transversal logical $R_Z(\theta)$ gates for an arbitrarily small $\theta$. We give an explicit recursive construction for such codes and calculate their parameters. Section \ref{Sec: Surface as another} broadens the applicability of this approach to the family of rotated surface codes and provides an overhead optimization method. 
In Section \ref{Sec: code switching}, we show that the code switching protocol can be extended to quantum codes with an arbitrary small phase rotation gate. Details on our circuit-level noise simulation is also provided in this section.

%%%%%%%%%%%%%%%%%%%%%%%%%%%%%%%%%%%%
\section{Preliminaries}\label{sec: preliminaries}
%%%%%%%%%%%%%%%%%%%%%%%%%%%%%%%%%%%%

%%%%%%%%%%%%%%%%%%%%%%%%%%%%%%%%%%
\subsection{Quantum Stabilizer Codes}
%%%%%%%%%%%%%%%%%%%%%%%%%%%%%%%%%%

The state space of a single qubit is the two-dimensional Hilbert space $\mathbb{C}^2$ with
computational basis $\{\ket{0},\ket{1}\}$. 
The state space of $n$ qubits is given by the
tensor product space $(\mathbb{C}^2)^{\otimes n}$. 
Errors acting on qubits can be decomposed as a linear combination of the Pauli operators
\[
\begin{split}
I =
\begin{pmatrix}
1 & 0 \\
0 & 1
\end{pmatrix}, \quad
&X =
\begin{pmatrix}
0 & 1 \\
1 & 0
\end{pmatrix}, \quad
Y =
\begin{pmatrix}
0 & -i \\
i & 0
\end{pmatrix}, \quad \\&
Z =
\begin{pmatrix}
1 & 0 \\
0 & -1
\end{pmatrix}.
\end{split}
\]
The $n$-qubit Pauli group $\mathcal{P}_n$ consists of all $n$-fold tensor products of
these operators, together with the phases $\{\pm 1, \pm i\}$.

A {\em quantum code} $Q$ encoding $k$ logical qubits into $n$ data qubits is a
$2^k$-dimensional subspace
\[
Q \subseteq (\mathbb{C}^2)^{\otimes n}.
\]
This code protects quantum information from errors affecting the physical qubits. 
In particular, the code is designed such that a specified set of error operators, typically drawn from the $n$-qubit Pauli group $\mathcal{P}_n$, can be detected
and corrected without disturbing the encoded logical information. 
Such a code is denoted by $[\![n,k]\!]$.

One of the most important and widely used approaches for constructing quantum codes is the stabilizer construction, which we review below. 
\begin{definition}[Stabilizer Code]\cite{Calderbank,Gottesman}
A {\em quantum (stabilizer) code} is defined by an
abelian subgroup $\mathcal{S} \subset \mathcal{P}_n$ that does not contain $-I$.
The quantum codespace is the simultaneous $+1$ eigenspace of all elements of $\mathcal{S}$, namely
\[
Q = \{ \ket{\psi} \in (\mathbb{C}^2)^{\otimes n} : S\ket{\psi}=\ket{\psi}
\text{ for all } S\in\mathcal{S} \}.
\]
If $\mathcal{S}$ is generated by $n-k$ independent commuting Pauli operators, then the
resulting stabilizer code is an $[\![n,k]\!]$ stabilizer code.
\end{definition}
We call an $n$-qubit unitary matrix $U$ a {\em logical operator (gate)} of $Q$, if $U\mathcal{S}U^\dagger\subseteq \bar{\mathcal{S}}$, where $\bar{\mathcal{S}}$ is the set of all Pauli and non-Pauli stabilizers of the code.
Let
\[
\mathcal{N}(\mathcal{S})=\{P\in\mathcal{P}_n \mid PS=SP \text{ for all } S\in\mathcal{S}\}
\]
denote the {\em normalizer group} of $\mathcal{S}$ in $\mathcal{P}_n$.
The \emph{minimum distance} $d$ of the stabilizer code is defined as
\[
d=\min\{\mathrm{wt}(P)\mid P\in\mathcal{N}(\mathcal{S})\setminus\mathcal{S}\},
\]
where $\mathrm{wt}(P)$ denotes the number of qubits on which the Pauli
operator $P$ acts nontrivially.

Operators in $\mathcal{N}(\mathcal{S})$ preserve the codespace. Indeed, elements of
$\mathcal{N}(\mathcal{S}) \setminus \mathcal{S}$ act nontrivially on the
codespace and correspond to logical Pauli operators acting on the encoded qubits.
For an $[\![n,k]\!]$ stabilizer code, one can choose $k$ pairs of logical Pauli
operators $\overline{X}_i$ and $\overline{Z}_i \in \mathcal{N}(\mathcal{S}) \setminus \mathcal{S}$ that satisfy the same
commutation relations as single-qubit Pauli operators.

If the minimum distance of a quantum code is known its parameters are shown as $[\![n,k,d]\!]$. This code can detect up to $d-1$ errors and can correct up to $\lfloor (d-1)/2 \rfloor$ arbitrary qubit errors.

A subclass of quantum stabilizer codes is the family of  Calderbank--Shor--Steane (CSS) codes \cite{CS96,steane1996multiple}. These are stabilizer codes whose stabilizer group can be
generated by operators that are either tensor products of $X$ operators or tensor
products of $Z$ operators. 
Equivalently, a CSS code can be described using two classical binary linear codes as stated below.

\begin{definition}[CSS Code]
Let $C_2\subseteq C_1 \subseteq \F_2^n$ be binary linear codes with $\dim(C_1)=k_1$ and $\dim(C_2)=k_2$. 
Then one can construct a quantum code with parameters $[\![n,k_1-k_2,d]\!]$, where 
\[
d=\min \{d_x=\wt(C_1\setminus C_2),d_z=\wt(C_2^\bot \setminus C_1^\bot)\},
\]
where $\mathrm{wt}()$ here denotes the minimum Hamming weight.
\end{definition}
Here the $X$-type and $Z$-type stabilizers are in correspondence to the binary codes $C_2$ and $C_1^\bot$, respectively. 
Moreover, one can determine the logical $X$ and $Z$ operators as elements of $C_1\setminus C_2$ and  $C_2^\bot \setminus C_1^\bot$, respective.
\begin{remark}
Although $C_1$ and $C_2$ are binary linear codes, 
with a slight abuse of terminology, we frequently refer to the elements of $C_2$ and $C_1^\perp$ as the $X$- and $Z$-stabilizers, respectively.
\end{remark}

Quantum color codes form an important family of topological quantum error-correcting codes introduced by Bombín and Martin-Delgado \cite{bombin2006topological,Color2,eczoo_2d_color}. 
They are defined on suitably colorable lattices and belong to the class of CSS stabilizer codes. 
One of their main advantages is that their geometric structure enables the implementation of several FT logical gates. 
We recall a description of them below.

\begin{definition}
A quantum color code is a CSS stabilizer code defined on a $(D+1)$-colorable $D$-dimensional lattice, where qubits are placed on vertices and $X$- and $Z$-stabilizers are associated with the cells of the lattice.
\end{definition}
We have provided more details about color codes and their description from the point of view of cell complexes in Appendix \ref{Sec: colandlogics}.

\subsection{Conditions for transversal logical phase gates}
In this section we briefly recall the necessary and sufficient conditions for a quantum CSS code to realize a logical gate transversally ({\em without requiring a correction operator}). 
Since we are only interested in codes with one logical qubit and only $Z$-rotation gates, we state the conditions tailored to such restrictions. Please refer to \cite{webster2023transversal,leitch2025transversal} for other general cases. 

A quantum logical operator is called {\em transversal} if it can be implemented as a tensor product of single-qubit gates acting on the data qubits. 
The main feature of transversality is that a single fault occurring during the implementation process cannot propagate to more than one data qubit, and such operation can be parallelized on a circuit with depth one.
A natural generalization of transversal gates is {\em $k$-fold transversality}, where a logical gate is implemented as a tensor product of $k$-qubit gates acting on disjoint subsets of the data qubits \cite{chakraborty2026no}.

For an integer $r\ge 0$, recall that a $Z$-rotation gate by the angle $\theta=\frac{\pi}{2^r}$ has the form 
\[
R_Z(\theta) = e^{-i \frac{\theta}{2} Z} = \begin{bmatrix} e^{-i\theta/2} & 0 \\ 0 & e^{i\theta/2} \end{bmatrix}=
\begin{bmatrix} 1 & 0 \\ 0 & e^{i\theta} \end{bmatrix}.
\]
Such a gate belongs to the $r+1$-th level of the Clifford hierarchy.
Recall also that the {\em Schur product} of two length $n$ binary vectors $u$ and $v$ is defined by 
\[
u \ast v=(u_1v_1,u_2v_2,\ldots,u_nv_n). 
\]
One can define the Schur product of more than two vectors analogously.
\begin{theorem}(Condition for Transversal Logical Phase gate \cite{webster2023transversal,leitch2025transversal}) \label{T:NandS for logical}
Let $Q$ be a quantum CSS code with the $X$-stabilizers generator set $H_X$ and $X$-logical operator $L_X=\{u\}$.
Then applying $R_Z(\frac{\pi}{2^{r-1}})^i$, for some odd integer $i$, to all data qubits implies a logical $R_Z(\frac{\pi}{2^{r-1}})$ if the following conditions are satisfied
\begin{itemize}
    \item $\wt(u)\equiv 1 \pmod 2$, $\wt(v)\equiv 0 \pmod{2^r}$ for each $v \in H_X$, and 
    \item Schur product of any  $2\le s \le r$ number of vectors in $H_X \cup L_X$ has weight divisible by $2^{r-s+1}$.
\end{itemize}
\end{theorem}
For instance, in a quantum divisible code, where each $X$-stabilizer has weight divisible by $2^r$, with the all ones as a logical operator, we have that the above conditions are satisfied.
Thus such a quantum code has a transversal logical $R_Z(\frac{\pi}{2^{r-1}})$ gate.

We have included a more general type of necessary and sufficient conditions for realizing a logical phase gate, transversally through a transversal application of different gates, in Appendix \ref{Sec: colandlogics}.

%%%%%%%%%%%%%%%%%%%%%%%%%%%%%%%%%%
\subsection{Triorthogonal codes and the doubling method}\label{Sec:triorthogonal}
%%%%%%%%%%%%%%%%%%%%%%%%%%%%%%%%%%
The doubling technique is a common, and perhaps one of the most resource-efficient currently-known methods in the literature for constructing quantum codes that enable the realization of logical 
$T$ gates via transversal physical gates \cite{berardini2025asymptotically,bravyi2015doubled,jain2025transversal}. 
The quantum triorthogonal codes constructed this way also form the main ingredients of the code switching protocol considered in \cite{heussen2025efficient}.

The main ingredients of the doubling construction are two quantum codes: a self-orthogonal (equivalently, a double-orthogonal code) code and a triorthogonal code of smaller distance. 
The resulting construction produces a triorthogonal quantum code with a larger distance. 
We next briefly recall the concepts of self-orthogonal and triorthogonal codes, as they form the fundamental objects used throughout this paper.

Recall that a binary linear code $C$ is called {\em even}, if each codeword of $C$ has an even (Hamming) weight.  
We call a quantum CSS code corresponding to $C_2\subseteq C_1$ {\em self-orthogonal}, if $C_2\subseteq C_1 =C_2^\bot$. 
Quantum self-orthogonal codes form an important class of quantum codes capable of realizing logical Clifford gates via multilevel transversal Clifford operations \cite{steane1999efficient,grassl2000cyclic,tansuwannont2025clifford}. 
There are several important families of quantum self-orthogonal codes, such as quadratic residue codes and duadic codes, and many of the currently best-known quantum codes are constructed using such codes \cite{GrasslT,aly2006remarkable,dastbasteh2024new,dastbasteh2025infinite,ketkar2006nonbinary,dastbasteh2024quantum}.

In a self-orthogonal quantum CSS code corresponding to $C_2 \subseteq C_1=C_2^{\bot}$, the $X$-stabilizers correspond to the vectors in $C_2$, all of which have even weight. 
If $C_2$ has only vectors of weights divisible by four, the code will be called  {\em doubly-even} and otherwise we call it {\em singly-even}. 

Quantum triorthogonal codes are a subclass of quantum self-orthogonal codes, and form the most general family of quantum codes that can realize the logical $T$ gate in a transversal manner (up to Pauli corrections) \cite{bravyi2012magic,rengaswamy2020optimality}.  

\begin{definition}[Triorthogonal Matrices  \cite{bravyi2012magic}]
We call a binary matrix $G$ of size $m \times n$ {\em triorthogonal} if and only if for each pair and each triple of its rows the following are satisfied:  
\[
\sum_{j=1}^nG_{aj}G_{b_j} \equiv 0\mod{2}\]
for all pairs of rows $1 \leq a < b \leq m$ and
\[
\sum_{j=1}^nG_{aj}G_{b_j}G_{c_j} \equiv 0\mod{2}
\]
for all triples of rows $1 \leq a < b < c\leq m.$
\end{definition}
The notion of triorthogonal matrices can be extended to $k$-orthogonal via a natural extension of the above definition. 

\begin{theorem}[Triorthogonal Codes \cite{bravyi2012magic}] \label{T: triorthogonal codes}
Let $G=[G_1,G_0]^T$ be a triorthogonal matrix, where $G_1$ consists of $k\geq0$ odd weight rows.
Let $[\![n,k]\!]$ be the CSS code of $C_2\subseteq C_1$, where $G_0$ and $G$ are generator matrices of the binary linear codes $C_2$ and $C_1$, respectively. 
Then applying transversal $T$ to the data qubits realizes the transversal logical $T$ (up to Pauli corrections) on the logical qubits.
\end{theorem}

The {\em doubling technique} is a method for constructing a larger quantum triorthogonal quantum code by combining a self-orthogonal code with a smaller triorthogonal code, effectively increasing the code distance \cite{betsumiya2012triply,bravyi2015doubled,jain2025transversal,berardini2025asymptotically}. 
In particular, it takes a quantum self-orthogonal code $Q_1$ with parameters $[\![n_1,1,d_1]\!]$ and a quantum triorthogonal code $Q_2$ with parameters $[\![n_2,1,d_2]\!]$ and returns a new quantum triorthogonal code with parameters
\[
[\![2n_1+n_2,1,\min\{d_1,d_2+2\}]\!],
\]
where the $X$-stabilizers of the new triorthogonal code are in the form of
\begin{equation}\label{E:general doubling matrix}
 \begin{bmatrix}
        E_1 & E_1&\bf{0_{n_2}}\\
        \bf{0_{n_1}}&\bf{0_{n_1}}&E_2\\
        \bf{0_{n_1}}&\bf{1_{n_1}}&\bf{1_{n_2}}
    \end{bmatrix},
\end{equation}
and $E_1$ and $E_2$ are generators of the $X$-stabilizers of $Q_1$ and $Q_2$, respectively. 
A simple representation, but not of the minimum weight, for the logical $X$-operator of the new code is $(\bf{1_{n_1}}, \bf{1_{n_1}},\bf{1_{n_2}})$.

A natural generalization of self-orthogonal and triorthogonal codes are {\em $k$-orthogonal codes}, where the overlap between any $k$ number of $X$-type Pauli stabilizer is even \cite{koutsioumpas2022smallest}.

%%%%%%%%%%%%%%%%%%%%%%%%%%%%%%%%%%%%%%%%%%%%%%%%%%%
\section{Quantum color codes with transversal logical $R_Z(\theta)$ gates}\label{Sec: general family}
%%%%%%%%%%%%%%%%%%%%%%%%%%%%%%%%%%%%%%%%%%%%%

The main goal of this section is to present a family of quantum codes
containing a single logical qubit that realize the logical $R_Z(\theta)$ transversally, for an arbitrary angle $\theta=\frac{\pi}{2^r}$. 
The family presented in this paper forms a subclass of color codes \cite{bombin2006topological,Color2,kubica2015universal,bombin2013introduction,eczoo_color}. 

One of the main features that singles out color codes from the other families of quantum codes is their capability of realize logical $R_Z(\theta)$ gates, as discussed in \cite{kubica2015universal,bombin2013introduction}. 
However, apart from certain numerical examples and a few scalable 2D and 3D families, the parameters and the systematic construction of such codes have not progressed much. 
In particular, the main scalable subclasses of color codes that realize logical $T$ gates transversally include 
\begin{enumerate}
    \item {\em Bravyi-Cross doubled color codes} \cite{bravyi2015doubled} 
    \[
    [\![\frac{d^3+5d^2-d-9}{4},1,d]\!]
    \]
\item {\em recursive capped color codes} \cite{tansuwannont2022achieving} 
\[
[\![\frac{d^3 + 3d^2 + 3d - 3}{4},1,d]\!],
\]
\item traditional 3D color codes \cite{bombin2015gauge} with parameters 
\[
[\![\frac{d^3 +d}{2},1,d]\!],
\]
\item and {\em Stacked codes} \cite{jochym2016stacked}
with parameters
\[
[\![\frac{3d^3 - 3d^2 + d + 3}{4},1,d]\!],
\]
\end{enumerate}
for an odd distance $d$.
To the best of our knowledge, explicit closed form formulas for the parameters of quantum color codes that realize finer 
$Z$-rotation gates while achieving increasing minimum distance have not been reported in the literature. 

\begin{table*}[t]
\centering
\small
\renewcommand{\arraystretch}{1.4}
\setlength{\tabcolsep}{6pt}
\begin{tabular}{l c c c}
\hline
Family & Parameters $[\![n,k,d]\!]$ & Logical gate & Reference \\
\hline

Our 2D family & $[\![\frac{d^2+2d-1}{2},1,d]\!]$ & Clifford & This work \\

Traditional 2D color codes & $[\![\frac{3(d^2-1)}{4}+1,1,d]\!]$ & Clifford & \cite{bombin2015gauge} \\

Capped color codes & $[\![\frac{3(d^2-1)}{2}+3,1,d]\!]$ & Clifford & \cite{tansuwannont2022achieving} \\

\hline

Our 3D family & $[\![\frac{d^3+6d^2+5d-6}{6},1,d]\!]$ & $T$ & This work \\

Bravyi-Cross doubled color codes & $[\![\frac{d^3+5d^2-d-9}{4},1,d]\!]$ & $T$ & \cite{bravyi2015doubled} \\

Traditional 3D color codes & $[\![\frac{d^3+d}{2},1,d]\!]$ & $T$ & \cite{bombin2015gauge} \\

Recursive capped color codes & $[\![\frac{d^3+3d^2+3d-3}{4},1,d]\!]$ & $T$ & \cite{tansuwannont2022achieving} \\

Stacked codes & $[\![\frac{3d^3-3d^2+d+3}{4},1,d]\!]$ & $T$ & \cite{jochym2016stacked} \\

\hline

Our $r$-dimensional family, $r,k \ge 2$ & $[\![\sum_{i=0}^{r} 2^i \binom{i+k-2}{i},1,2k-1]\!]$ & $R_Z(\frac{\pi}{2^{r-1}})$ & This work \\
\hline
\end{tabular}
\caption{Comparison of the qubit overhead of our code families with other color codes and the logical gates they enable transversally.}
\label{tab: comparison}
\end{table*}

In this section, we present a family of $2^r$-divisible quantum color codes with parameters 
\[
[\![S_r(k),1,2k-1]\!]
\]
that realizes logical $\frac{\pi}{2^{r-1}}$-$Z$ rotation gate transversally, where
\[S_r(k)=\sum_{i=0}^{r} 2^i \binom{i+k-2}{i}.
\] 
In particular, the 2D and 3D subclasses of our code family, which realize logical Clifford and $T$ gates, respectively, require less overhead than the families 1--4 reported above. 
Interestingly, optimal self-orthogonal and triorthogonal codes with small distances such as $[\![7,1,3]\!], [\![17,1,5]\!],[\![15,1,3]\!]$, and $[\![49,1,5]\!]$ appear as the first instances of color codes inside our family (see Table \ref{tab: comparison} for details). 
Our proofs employ elementary tools from linear algebra to describe these quantum codes, enabling the exact calculation of code parameters as well as their explicit stabilizers and logical operators.

Another feature of our approach is that it enables us to extend the idea to other classes of codes with different geometries including rotated surface codes. 
Moreover, we take advantage of the structure of stabilizer operators to identify $Z$ meta-check operators that enable single shot decoding of $Z$-syndromes.

We finish this part by highlighting that all the distance $d=3$ codes presented for the rest of this section are optimal; in fact, they have the minimum length among all color codes with this property, as they coincide with the codes presented in \cite{koutsioumpas2022smallest}.

%%%%%%%%%%%%%%%%%%%%%%%%%%%%%%%%%%%%%%%%%%%%%%%%
\subsection{Quantum self-orthogonal codes via doubling}
%%%%%%%%%%%%%%%%%%%%%%%%%%%%%%%%%%%%%%%%%%%%%%%%%%%%%
In this section, we first reveal an application of the doubling method in generating quantum self-orthogonal codes. 
Next, we show that well-known examples of self-orthogonal codes, including the $[\![7,1,3]\!]$ Steane and the $[\![17,1,5]\!]$ color codes, can be constructed in this way. 
Then, we extend the construction to build a family of 2D quantum (color) codes that are degenerate and many of their stabilizer generators have {\em weight four}. 

We start by formalizing our main tool of this work, namely doubling technique as a method of combining codes and its properties.  
\begin{theorem} \label{T:Doubling-general}
 Let $Q_1$ and $Q_2$ be $[\![n_1,1,d_1]\!]$ and $[\![n_2,1,d_2]\!]$ quantum codes of odd lengths with the all ones vector as a logical X-operator. Then \begin{enumerate}

     \item There exists a quantum code $Q$ with parameters $[\![2n_1+n_2,1,\min\{(d_1)_z,(d_2)_z+2\}]\!]$. 
\item The $X$-stabilizers of $Q$ are in the form
\begin{equation*}
G= \begin{bmatrix}
        E_1 & E_1&\bf{0_{n_2}}\\
        \bf{0_{n_1}}&\bf{0_{n_1}}&E_2\\
        \bf{0_{n_1}}&\bf{1_{n_1}}&\bf{1_{n_2}}
    \end{bmatrix},
\end{equation*}
where $E_1$ and $E_2$ are the generators of $X$-stabilizers of $Q_1$ and $Q_2$, respectively. 

\item The code $Q$ has an $X$-logical operator of the form $(\bf{1_{n_1}}, \bf{0_{n_1}},\bf{0_{n_2}})$.  
 \end{enumerate}
Furthermore, if $2^{r-1}$ and $2^{r}$ divide the weight of all $X$-stabilizers in $Q_1$ and $Q_2$, respectively, then $Q$ realizes a transversal $\frac{\pi}{2^{r-1}}$-$Z$ rotation logical gate.
\end{theorem}

The proof follows the same ideas of the doubling construction of triorthogonal codes presented in \cite{bravyi2015doubled,jain2025transversal, berardini2025asymptotically}, and \cite[Section 5.2]{hu2022designing}. 
A sketch of the proof is provided in Appendix \ref{A:closed formula}. 
 
Through simple linear algebra arithmetic, one one can show that the matrix $G$ of the above theorem has an equivalent (but optimal) representation of  
\begin{equation}\label{E:doubling matrix2}
G = \begin{bmatrix}
        E_1 & E_1&\bf{0_{n_2}}\\
        \bf{0_{n_1}}&\bf{0_{n_1}}&E_2\\
        \bf{0_{n_1}}&\bf{1_{n_1}}&\bf{v}
    \end{bmatrix},
\end{equation}
where $v$ is in correspondence to a minimum weight logical operator of the quantum code $Q_2$, and $E_1$ and $E_2$ have minimum weight stabilizers. 

As the first application, we substitute $Q_1$ and $Q_2$ in the previous theorem with the all-even-code of size $d+2$ and a quantum self-orthogonal code of an odd distance $d$. 
In other words, the next theorem enables to increase the distance of a quantum self-orthogonal code from $d$ to $d+2$ at a cost of $2(d+2)$ extra data qubits. 

\begin{corollary}\label{T:Doubling-selforthogonal}
Let $Q$ be a binary self-orthogonal quantum CSS code of odd length $n$ corresponding to $C_2 \subseteq C_1$, where $C_1=C_2\oplus\Span\{\bf{1_{n}}\}$, and with the parameters $[\![n,1,d]\!]$. Then 
\begin{enumerate}
    \item there exists a quantum self-orthogonal code $Q'$ with parameters
\[
[\![n+2(d+2), 1, d+2]\!].\]
\item the X- (or Z-)stabilizers of this code can be chosen to 
\begin{itemize}
    \item $d+1$ of them have weight four, 
    \item one with weight $2d+2$, 
    \item and the rest have the same weight as a generator set for $X$-stabilizers of the code $Q$.
    \end{itemize}
\end{enumerate}
\end{corollary}

\begin{proof}
The proof is an immediate application of Theorem \ref{T:Doubling-general}. 
We just note that such a new quantum code has $X$-stabilizer generators corresponding to the matrix $G$ of the form \eqref{E:doubling matrix2} by choosing $E_1$ to be the generator matrix of the all-even code and $E_2$ be the lowest weight generators of $C_2$. 
The orthogonality of the rows of $G$ and the weight of its stabilizers can be checked by inspecting $G$. 
$\hfill \square$ \end{proof}

It should be noted that if $d>2$ in Corollary \ref{T:Doubling-selforthogonal}, then the new quantum code $Q'$ obtained through the above result is degenerate to four (the minimum distance of corresponding classical code is four). 
It is because the first $d+1$ rows of $G$ are in corresponding to stabilizers with weight four. However, the minimum distance of $Q$ is at least five. 

The singly-even or doubly-even characteristic of the quantum code $Q'$ of Corollary \ref{T:Doubling-selforthogonal} depends only on the code $Q$. In particular, $Q'$ is doubly-even if and only if $Q$ is doubly-even. 
This is because stabilizers of $Q'$ are in correspondence to the matrix $G$ presented in (\ref{E:doubling matrix2}), and 
except the middle block, which is in correspondence to stabilizer generators of $Q$, all the other blocks have weights divisible by four.

Before stating our numerical examples, we would like to point out the smallest code that is conventionally $k$-orthogonal for each $k\ge 1$. 
As we will see later, this code plays a significant role in the structure of many important codes including $[\![7,1,3]\!]$ Steane code, $[\![15,1,1]\!]$ Reed-Muller code, and many other codes.

\begin{remark}[The totally orthogonal code]
The smallest $k$-ortohognal code for each $k\ge 1$ is the $[\![1,1,1]\!]$ CSS code that has no $X$ and $Z$ stabilizers. The corresponding CSS code is $C_2\subseteq C_1$, where $C_2=\{(0)\}$ and $C_{1}=\{(0),(1)\}$. This code will be called the {\em totally orthogonal} code.    
\end{remark}

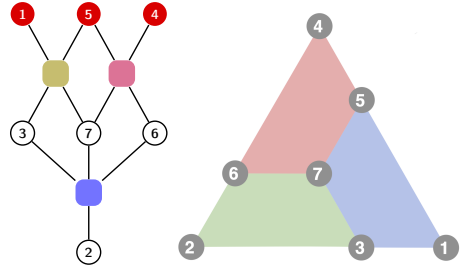
\begin{figure}
    \centering
\scalebox{0.35}{
\begin{tikzpicture}[
    var node/.style={
        circle,
        draw=black,
        line width=1.5pt,
        minimum size=9mm,
        inner sep=0pt,
        font=\sffamily\bfseries\Large
    },
    red node/.style={
        var node,
        fill=red!85!black,
        draw=none,
        text=white
    },
    check node/.style={
        rectangle,
        rounded corners=3mm,
        fill=cyan!50!white, 
        minimum size=10mm,
        inner sep=0pt
    },
    edge/.style={
        draw=black,
        line width=1.5pt
    }
]    
    \node[red node] (n1) at (-2.5, 4.5) {1};
    \node[red node] (n5) at (0, 4.5)    {5};
    \node[red node] (n4) at (2.5, 4.5)  {4};
    \node[check node,olive!55] (b1) at (-1.25, 2.25) {};
    \node[check node, purple!55] (b2) at (1.25, 2.25)  {};
    \node[var node] (n3) at (-2.5, 0) {3};
    \node[var node] (n7) at (0, 0)    {7};
    \node[var node] (n6) at (2.5, 0)  {6};
    \node[check node,blue!55] (b3) at (0, -2.25) {};
    \node[var node] (n2) at (0, -4.5) {2};
    \path[edge] (n1) -- (b1);
    \path[edge] (n5) -- (b1);
    \path[edge] (n3) -- (b1);
    \path[edge] (n7) -- (b1);
    \path[edge] (n5) -- (b2);
    \path[edge] (n4) -- (b2);
    \path[edge] (n7) -- (b2);
    \path[edge] (n6) -- (b2);
    \path[edge] (n3) -- (b3);
    \path[edge] (n7) -- (b3);
    \path[edge] (n6) -- (b3);
    \path[edge] (n2) -- (b3);

\end{tikzpicture}}
    \includegraphics[width=0.5\linewidth]{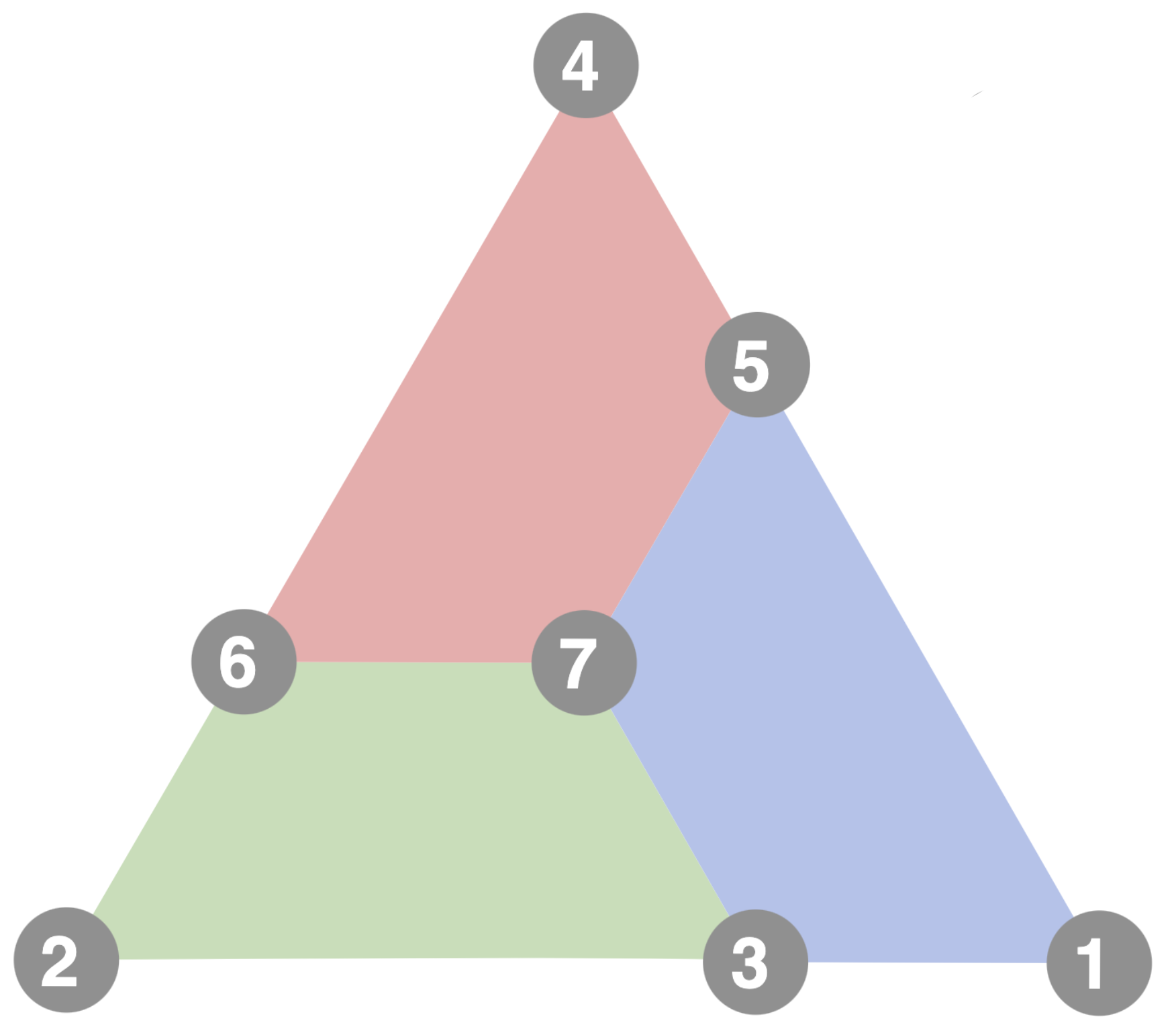}
    \caption{The regular representation of seven qubits color/Steane code (right) and the one obtained from Corollary \ref{T:Doubling-selforthogonal} (left). On the left figure, qubits 1,5,4 (respectively 3,7,6) form all-even  code, and qubit 2 is the totally orthogonal code. Here the logical X operator is $X_{1}X_{5}X_{4}$ (top row of qubits).}
    \label{fig:SteaneCode}
\end{figure}

Through the following example, we show that the $[\![7,1,3]\!]$ Steane and the $[\![15,1,3]\!]$ Reed-Muller codes can be  constructed as an application of Corollary \ref{T:Doubling-selforthogonal}. 

\begin{example}\label{example: higherd=3}
(I) Let $Q$ be the $[\![1,1,1]\!]$ totally orthogonal code and $E_1$ be the generator matrix of all-even  code of length $3$. Applying the result of Corollary \ref{T:Doubling-selforthogonal} gives the $[\![7,1,3]\!]$ Steane code (see Figure \ref{fig:SteaneCode}). 

(II) Applying the result of Theorem \ref{T:Doubling-general} to Steane code (doubly even) and the quantum totally orthogonal code implies the quantum  $[\![15,1,3]\!]$ Reed-Muller code. 

(III) Repeating the same process
to $[\![15,1,3]\!]$ and $[\![1,1,1]\!]$ gives the $[\![31,1,3]\!]$ quantum code which is $4$-orthogonal. 
Consequently redoing the above process produces the smallest minimum distance three quantum codes with transversal logical $\frac{\pi}{2^{r-1}}$-$Z$ rotation gate, as discussed in \cite{koutsioumpas2022smallest}, namely quantum simplex (or Hamming) $[\![2^r-1,1,3]\!]$ code for each $r \ge 3$.
\end{example}

The above example showed how to build quantum codes with transversal logical operators belonging to {\em level $r\ge 2$ of the Clifford hierarchy} via using the result of Corollary \ref{T:Doubling-selforthogonal} once, and then applying Theorem \ref{T:Doubling-general} to the resulting code (see Figure \ref{fig:CHd3}). 
Later we extend this process to reach any other Clifford level for each desired odd minimum distance.    

The next example shows another application in constructing a self-orthogonal code with parameters $[\![41,1,9]\!]$ which is also doubly even and has less number of data qubits than the doubly even $[\![45,1,9]\!]$ code with distance $9$ discussed in \cite{jain2025transversal} or \cite[Table 1]{ouyang2025measurement}.  

\begin{example}
Let $n=23$ and $C_2$ be the doubly even linear cyclic code with the generator polynomial $x^{12} + x^{10} + x^7 + x^4 + x^3 + x^2 + x + 1
$. The code $C_2$ is self-orthogonal, i.e., $C_2 \subseteq C_1=C_2^\bot$, and has parameters $[23,11,8]$ and the corresponding quantum CSS code of $C_2\subseteq C_1$ is the non-degenerate $[\![23,1,7]\!]$. 
Let $C_2'$ be the all-even  code of length $9$ which has parameters $[9,8,2]$. Applying the doubling construction of Corollary \ref{T:Doubling-selforthogonal} implies the existence of a degenerate quantum code $Q$ with parameters $[\![41,1,9]\!]$ (degenerate to 4). 
Note also that the $X$ (or $Z$) stabilizers of the latter quantum code are in correspondence to the matrix $G$ presented in (\ref{E:doubling matrix2}), and since $C_2$ is doubly even, one can conclude that all the rows of $G$ have weight divisible by four. Therefore, the code generated by $G$ is also doubly even. 

Another round of the doubling method based on the quantum code $Q$ above and the triorthogonal (triply even) code $[\![95,1,7]\!]$ implies a triorthogonal (triply even) code with parameters $[\![177,1,9]\!]$. 
The latter code has lower length (data qubits)  than the distance $9$ triorthogonal code presented in \cite[Table I]{jain2025transversal}. 
Using the resulting code and repeating the above argument, one can improve the parameters of several other triorthogonal codes in \cite[Table I]{jain2025transversal}. In particular, the distance $9$ -- $19$ triorthogonal codes with the improved lengths are 
$[\![ 177, 1, 9 ]\!],
[\![ 271, 1, 11 ]\!],
[\![ 409, 1, 13 ]\!],
[\![ 567, 1, 15 ]\!]
$, $[\![ 769, 1, 17 ]\!],
$ and $[\![ 975, 1, 19 ]\!]$. 
\end{example}

\begin{figure}
    \centering
\scalebox{0.6}{
\begin{tikzpicture}[>=stealth, scale=1.2]
    \draw[->, thick] (0,0) -- (7,0) node[right] {$d$ (odd)};
    \draw[->, thick] (0,0) -- (0,6) node[above] {$\mathcal{C}^{(k)}$};
    \foreach \x/\label in {2/1, 4/3, 6/5}
        \draw (\x, 0.1) -- (\x, -0.1) node[below] {\label};
    \foreach \y/\label in {1/2, 2.5/3, 4/4, 5.5/k}
        \draw (0.1, \y) -- (-0.1, \y) node[left] {$\mathcal{C}^{(\label)}$};
    \draw[dashed, gray!50] (4,0) -- (4,6);
    \draw[->, ultra thick, red] (4,1) -- (4,2.5) node[midway, left] {Example \ref{example: higherd=3}};
    \filldraw[black] (4,1) circle (2.5pt) node[right, xshift=3pt] {$\mathbf{(3, \mathcal{C}^{(2)})}$};
    \filldraw[black] (4,2.5) circle (2.5pt) node[right, xshift=3pt] {$\mathbf{(3, \mathcal{C}^{(3)})}$};
    \filldraw[black] (4,4) circle (2.5pt) node[right, xshift=3pt] {$\mathbf{(3, \mathcal{C}^{(4)})}$};
    
    \filldraw[black] (4,4.5) circle (1pt);
    \filldraw[black] (4,4.75) circle (1pt);
    \filldraw[black] (4,5.0) circle (1pt);
    \filldraw[black] (-0.5,4.5) circle (1pt);
    \filldraw[black] (-0.5,4.75) circle (1pt);
    \filldraw[black] (-0.5,5.0) circle (1pt);
    \filldraw[black] (4,5.5) circle (2.5pt) node[right, xshift=3pt] {$\mathbf{(3, \mathcal{C}^{(k)})}$};

\end{tikzpicture}
}
\caption{Distance three quantum codes with logical gates from an arbitrary Clifford hierarchy through the process of Example \ref{example: higherd=3}. }
\label{fig:CHd3}
\end{figure}
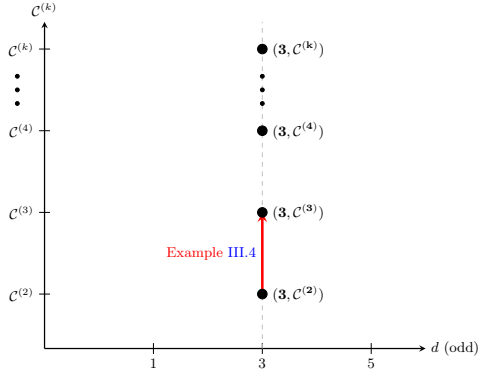

For now, we restrict our attention only to the all-even  code, and use the result of Corollary \ref{T:Doubling-selforthogonal} recursively. 
This way one gets the following family of self-orthogonal codes (belonging to the second level of the Clifford hierarchy) that can have any desired odd distance. 
This family has a 2D (triangular) representation, and one can verify that they admit a three-valent three-colorable tiling geometry. 
Thus they belong to the family of 2D color codes. The qubit connectivity can be easily viewed through a recursive pattern (see Figure \ref{fig:recursive}).

\begin{theorem}\label{T:Doubling color code1}
Let $k$ be a positive integer. 
There exists a family of  self-orthogonal (2D color) codes with parameters
\[
[\![2k^2-1,1,2k-1]\!].
\]
If $k\ge 2$, this quantum code is degenerate, and has a generator set of $X$ (or $Z$) stabilizers, where 
\begin{itemize}
    \item $k^2-k-3$ of the generators have weight four,  
    \item and $k-2$ of them that have weights
\[8,12,16,\ldots,4k-4.\]
\end{itemize}
In particular, all the codes inside this family are doubly even.\\ 
This quantum code has a 2D triangular   representation, where all the data qubits are partitioned in
\begin{itemize}
    \item 2 layers of $2s+1$ data qubits, for each $s=1,\ldots,k-1$, in a descending order, and one single data qubit at the bottom,
    \item each such two layers contain $s-1$ check qubits of weight four,
    \item the bottom layer of $2s+1$ and top layer of $2s-1$ data qubits are all connected through a check qubit,
    \item the highest layer (containing $2k-1$ data qubits) represents the location of a minimum weight logical operator.  
\end{itemize} 
Finally, the desired logical $S$ gate can be realized by the physical action of $S$ on the odd number layers of qubits, and $S^\dagger$ on the even number layers of qubits.  
\end{theorem}
A complete proof of this theorem is given in Appendix \ref{A:closed formula}.

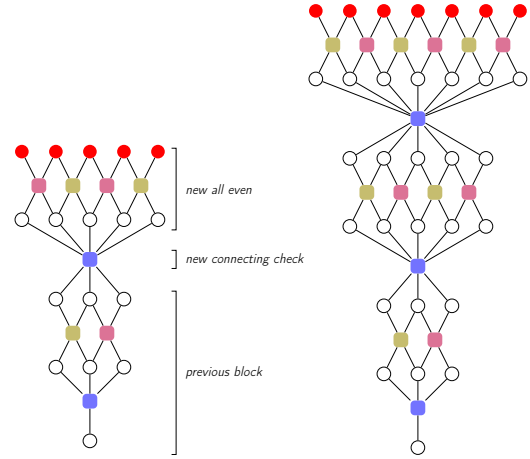
\begin{figure}
    \centering
\scalebox{0.3}{
\begin{tikzpicture}[
    white node/.style={
        circle,
        draw=black,
        line width=1.2pt,
        minimum size=6mm,
        fill=white,
        inner sep=0pt
    },
    red node/.style={
        circle,
        fill=red,
        minimum size=6mm,
        inner sep=0pt
    },
    check node/.style={
        rectangle,
        rounded corners=1.8mm,
        fill=cyan!60!white,
        minimum size=6.5mm,
        inner sep=0pt
    },
    edge/.style={
        draw=black,
        line width=1.1pt
    },
    label style/.style={
        font=\sffamily\itshape\LARGE,
        anchor=west
    }
]
    \node[white node] (p_bot) at (0, -4.5) {};
    \node[check node,blue!55] (p_c3)  at (0, -2.75) {};
    
    \node[white node] (p_w3)  at (-1.5, -1.25) {};
    \node[white node] (p_w7)  at (0, -1.25) {};
    \node[white node] (p_w6)  at (1.5, -1.25) {};
    
    \node[check node,olive!55] (p_c1)  at (-0.75, 0.25) {};
    \node[check node,purple!55] (p_c2)  at (0.75, 0.25) {};
    
    \node[white node] (p_t3)  at (-1.5, 1.75) {};
    \node[white node] (p_t7)  at (0, 1.75) {};
    \node[white node] (p_t6)  at (1.5, 1.75) {};
    \path[edge] (p_bot) -- (p_c3);
    \path[edge] (p_c3)  -- (p_w3);
    \path[edge] (p_c3)  -- (p_w7);
    \path[edge] (p_c3)  -- (p_w6);
    \path[edge] (p_w3)  -- (p_c1);
    \path[edge] (p_w7)  -- (p_c1);
    \path[edge] (p_w7)  -- (p_c2);
    \path[edge] (p_w6)  -- (p_c2);
    \path[edge] (p_c1)  -- (p_t3);
    \path[edge] (p_c1)  -- (p_t7);
    \path[edge] (p_c2)  -- (p_t7);
    \path[edge] (p_c2)  -- (p_t6);
    \node[check node,blue!55] (mid_check) at (0, 3.5) {};

    \path[edge] (p_t3) -- (mid_check);
    \path[edge] (p_t7) -- (mid_check);
    \path[edge] (p_t6) -- (mid_check);

    \node[white node] (u_w1) at (-3.0, 5.25) {};
    \node[white node] (u_w2) at (-1.5, 5.25) {};
    \node[white node] (u_w3) at (0, 5.25)    {};
    \node[white node] (u_w4) at (1.5, 5.25)  {};
    \node[white node] (u_w5) at (3.0, 5.25)  {};

    \node[check node,purple!55] (u_c1) at (-2.25, 6.75) {};
    \node[check node,olive!55] (u_c2) at (-0.75, 6.75) {};
    \node[check node,purple!55] (u_c3) at (0.75, 6.75)  {};
    \node[check node,olive!55] (u_c4) at (2.25, 6.75)  {};

    \node[red node] (u_r1) at (-3.0, 8.25) {};
    \node[red node] (u_r2) at (-1.5, 8.25) {};
    \node[red node] (u_r3) at (0, 8.25)    {};
    \node[red node] (u_r4) at (1.5, 8.25)  {};
    \node[red node] (u_r5) at (3.0, 8.25)  {};

    \path[edge] (mid_check) -- (u_w1);
    \path[edge] (mid_check) -- (u_w2);
    \path[edge] (mid_check) -- (u_w3);
    \path[edge] (mid_check) -- (u_w4);
    \path[edge] (mid_check) -- (u_w5);
    
    \path[edge] (u_w1) -- (u_c1);
    \path[edge] (u_w2) -- (u_c1);
    \path[edge] (u_w2) -- (u_c2);
    \path[edge] (u_w3) -- (u_c2);
    \path[edge] (u_w3) -- (u_c3);
    \path[edge] (u_w4) -- (u_c3);
    \path[edge] (u_w4) -- (u_c4);
    \path[edge] (u_w5) -- (u_c4);
    
    \path[edge] (u_c1) -- (u_r1);
    \path[edge] (u_c1) -- (u_r2);
    \path[edge] (u_c2) -- (u_r2);
    \path[edge] (u_c2) -- (u_r3);
    \path[edge] (u_c3) -- (u_r3);
    \path[edge] (u_c3) -- (u_r4);
    \path[edge] (u_c4) -- (u_r4);
    \path[edge] (u_c4) -- (u_r5);
    \draw[line width=1pt] (3.6, 8.4) -- (3.8, 8.4) -- (3.8, 4.8) -- (3.6, 4.8);
    \node[label style] at (4.1, 6.6) {new all even};
    \draw[line width=1pt] (3.6, 3.9) -- (3.8, 3.9) -- (3.8, 3.1) -- (3.6, 3.1);
    \node[label style] at (4.1, 3.5) {new connecting check};

    \draw[line width=1pt] (3.6, 2.1) -- (3.8, 2.1) -- (3.8, -5.1) -- (3.6, -5.1);
    \node[label style] at (4.1, -1.5) {previous block};
\end{tikzpicture}
\begin{tikzpicture}[
    white node/.style={
        circle,
        draw=black,
        line width=1.2pt,
        minimum size=6mm,
        fill=white,
        inner sep=0pt
    },
    red node/.style={
        circle,
        fill=red,
        minimum size=6mm,
        inner sep=0pt
    },
    check node/.style={
        rectangle,
        rounded corners=1.8mm,
        fill=cyan!60!white,
        minimum size=6.5mm,
        inner sep=0pt
    },
    edge/.style={
        draw=black,
        line width=1.1pt
    }
]
    \node[white node] (p_bot) at (0, -4.5) {};
    \node[check node,blue!55] (p_c3)  at (0, -2.75) {};
    
    \node[white node] (p_w3)  at (-1.5, -1.25) {};
    \node[white node] (p_w7)  at (0, -1.25) {};
    \node[white node] (p_w6)  at (1.5, -1.25) {};
    
    \node[check node,olive!55] (p_c1)  at (-0.75, 0.25) {};
    \node[check node,purple!55] (p_c2)  at (0.75, 0.25) {};
    
    \node[white node] (p_t3)  at (-1.5, 1.75) {};
    \node[white node] (p_t7)  at (0, 1.75) {};
    \node[white node] (p_t6)  at (1.5, 1.75) {};

    \path[edge] (p_bot) -- (p_c3);
    \path[edge] (p_c3)  -- (p_w3);
    \path[edge] (p_c3)  -- (p_w7);
    \path[edge] (p_c3)  -- (p_w6);
    
    \path[edge] (p_w3)  -- (p_c1);
    \path[edge] (p_w7)  -- (p_c1);
    \path[edge] (p_w7)  -- (p_c2);
    \path[edge] (p_w6)  -- (p_c2);
    
    \path[edge] (p_c1)  -- (p_t3);
    \path[edge] (p_c1)  -- (p_t7);
    \path[edge] (p_c2)  -- (p_t7);
    \path[edge] (p_c2)  -- (p_t6);
    
    \node[check node,blue!55] (mid_check) at (0, 3.5) {};
    \path[edge] (p_t3) -- (mid_check);
    \path[edge] (p_t7) -- (mid_check);
    \path[edge] (p_t6) -- (mid_check);
    \node[white node] (m_w1) at (-3.0, 5.25) {};
    \node[white node] (m_w2) at (-1.5, 5.25) {};
    \node[white node] (m_w3) at (0, 5.25)    {};
    \node[white node] (m_w4) at (1.5, 5.25)  {};
    \node[white node] (m_w5) at (3.0, 5.25)  {};

    \path[edge] (mid_check) -- (m_w1);
    \path[edge] (mid_check) -- (m_w2);
    \path[edge] (mid_check) -- (m_w3);
    \path[edge] (mid_check) -- (m_w4);
    \path[edge] (mid_check) -- (m_w5);

    \node[check node,olive!55]  (m_c1) at (-2.25, 6.75) {};
    \node[check node,purple!55] (m_c2) at (-0.75, 6.75) {};
    \node[check node,olive!55]  (m_c3) at (0.75, 6.75)  {};
    \node[check node,purple!55] (m_c4) at (2.25, 6.75)  {};

    \path[edge] (m_w1) -- (m_c1);
    \path[edge] (m_w2) -- (m_c1);
    \path[edge] (m_w2) -- (m_c2);
    \path[edge] (m_w3) -- (m_c2);
    \path[edge] (m_w3) -- (m_c3);
    \path[edge] (m_w4) -- (m_c3);
    \path[edge] (m_w4) -- (m_c4);
    \path[edge] (m_w5) -- (m_c4);

    \node[white node] (m_t1) at (-3.0, 8.25) {};
    \node[white node] (m_t2) at (-1.5, 8.25) {};
    \node[white node] (m_t3) at (0, 8.25)    {};
    \node[white node] (m_t4) at (1.5, 8.25)  {};
    \node[white node] (m_t5) at (3.0, 8.25)  {};
    \path[edge] (m_c1) -- (m_t1);
    \path[edge] (m_c1) -- (m_t2);
    \path[edge] (m_c2) -- (m_t2);
    \path[edge] (m_c2) -- (m_t3);
    \path[edge] (m_c3) -- (m_t3);
    \path[edge] (m_c3) -- (m_t4);
    \path[edge] (m_c4) -- (m_t4);
    \path[edge] (m_c4) -- (m_t5);
    
    \node[check node,blue!55] (high_check) at (0, 10.0) {};

    \path[edge] (m_t1) -- (high_check);
    \path[edge] (m_t2) -- (high_check);
    \path[edge] (m_t3) -- (high_check);
    \path[edge] (m_t4) -- (high_check);
    \path[edge] (m_t5) -- (high_check);

    \node[white node] (u_w1) at (-4.5, 11.75) {};
    \node[white node] (u_w2) at (-3.0, 11.75) {};
    \node[white node] (u_w3) at (-1.5, 11.75) {};
    \node[white node] (u_w4) at (0, 11.75)    {};
    \node[white node] (u_w5) at (1.5, 11.75)  {};
    \node[white node] (u_w6) at (3.0, 11.75)  {};
    \node[white node] (u_w7) at (4.5, 11.75)  {};

    \path[edge] (high_check) -- (u_w1);
    \path[edge] (high_check) -- (u_w2);
    \path[edge] (high_check) -- (u_w3);
    \path[edge] (high_check) -- (u_w4);
    \path[edge] (high_check) -- (u_w5);
    \path[edge] (high_check) -- (u_w6);
    \path[edge] (high_check) -- (u_w7);
    
    \node[check node,olive!55]  (u_c1) at (-3.75, 13.25) {};
    \node[check node,purple!55] (u_c2) at (-2.25, 13.25) {};
    \node[check node,olive!55]  (u_c3) at (-0.75, 13.25) {};
    \node[check node,purple!55] (u_c4) at (0.75, 13.25)  {};
    \node[check node,olive!55]  (u_c5) at (2.25, 13.25)  {};
    \node[check node,purple!55] (u_c6) at (3.75, 13.25)  {};
    
    \path[edge] (u_w1) -- (u_c1);
    \path[edge] (u_w2) -- (u_c1);
    \path[edge] (u_w2) -- (u_c2);
    \path[edge] (u_w3) -- (u_c2);
    \path[edge] (u_w3) -- (u_c3);
    \path[edge] (u_w4) -- (u_c3);
    \path[edge] (u_w4) -- (u_c4);
    \path[edge] (u_w5) -- (u_c4);
    \path[edge] (u_w5) -- (u_c5);
    \path[edge] (u_w6) -- (u_c5);
    \path[edge] (u_w6) -- (u_c6);
    \path[edge] (u_w7) -- (u_c6);
    
    \node[red node] (u_r1) at (-4.5, 14.75) {};
    \node[red node] (u_r2) at (-3.0, 14.75) {};
    \node[red node] (u_r3) at (-1.5, 14.75) {};
    \node[red node] (u_r4) at (0, 14.75)    {};
    \node[red node] (u_r5) at (1.5, 14.75)  {};
    \node[red node] (u_r6) at (3.0, 14.75)  {};
    \node[red node] (u_r7) at (4.5, 14.75)  {};
    
    \path[edge] (u_c1) -- (u_r1);
    \path[edge] (u_c1) -- (u_r2);
    \path[edge] (u_c2) -- (u_r2);
    \path[edge] (u_c2) -- (u_r3);
    \path[edge] (u_c3) -- (u_r3);
    \path[edge] (u_c3) -- (u_r4);
    \path[edge] (u_c4) -- (u_r4);
    \path[edge] (u_c4) -- (u_r5);
    \path[edge] (u_c5) -- (u_r5);
    \path[edge] (u_c5) -- (u_r6);
    \path[edge] (u_c6) -- (u_r6);
    \path[edge] (u_c6) -- (u_r7);

\end{tikzpicture}

}
    \caption{Left: recursive construction of $[\![17,1,5]\!]$ from $[\![7,1,3]\!]$ code of Figure \ref{fig:SteaneCode} (appear as the previous block). 
    Here three colored squares represent check qubits ($X$ and $Z$), white circles data qubits, and highlighted red circles a minimum weight logical operator ($X$ and $Z$).
    Right: recursive construction of $[\![31,1,7]\!]$ from $[\![17,1,5]\!]$ code.}
    \label{fig:recursive}
\end{figure}

An instance of such recursive pattern is depicted in Figure \ref{fig:recursive}, where it describes the quantum $[\![17,1,5]\!]$ self-orthogonal code. 
The 2D structure is trivial and the 3-coloring of vertices is shown in the figure. 
The locations of a minimum weight logical operator is shown in red. 
The picture also gives useful information about the weight of the stabilizer generators. 

The first two example of codes inside the family of Theorem \ref{T:Doubling color code1} are the well-known $[\![7,1,3]\!]$ and $[\![17,1,5]\!]$ 2D color codes \cite{bombin2006topological,eczoo_2d_color}.

The result of this theorem allows to construct codes that realize logical transversal gates for a generating Clifford set and for any desired odd minimum distance. 
This result allows us to realize certain points in the 2D  distance-Clifford hierarchy lattice (see Figure \ref{fig:CH2}). 

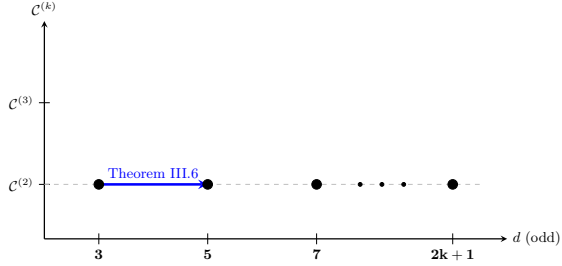
\begin{figure}[h]
\scalebox{0.6}{
\begin{tikzpicture}[>=stealth, scale=1.2]
    \draw[->, thick] (0,0) -- (8.5,0) node[right] {$d$ (odd)};
    \draw[->, thick] (0,0) -- (0,4) node[above] {$\mathcal{C}^{(k)}$};

    \foreach \x/\label in {1/3, 3/5, 5/7, 7.5/2k+1}
        \draw (\x, 0.1) -- (\x, -0.1) node[below] {$\mathbf{\label}$};

    \foreach \y/\label in {1/2, 2.5/3}
        \draw (0.1, \y) -- (-0.1, \y) node[left] {$\mathcal{C}^{(\label)}$};
    \draw[dashed, gray!50] (0,1) -- (8,1);
    \draw[->, ultra thick, blue] (1,1) -- (3,1) node[midway, above] {Theorem \ref{T:Doubling color code1}};
    \filldraw[black] (1,1) circle (2.5pt);
    \filldraw[black] (3,1) circle (2.5pt);
    \filldraw[black] (5,1) circle (2.5pt);
    \filldraw[black] (5.8,1) circle (1pt);
    \filldraw[black] (6.2,1) circle (1pt);
    \filldraw[black] (6.6,1) circle (1pt);
    \filldraw[black] (7.5,1) circle (2.5pt);

\end{tikzpicture}
}
\caption{The construction described in Theorem \ref{T:Doubling color code1} that allows to increase the distance of self-orthogonal quantum code with the aid of all-even codes.}
\label{fig:CH2}
\end{figure}

The exact transversal operator for realizing logical $S$ gate for such codes can be obtained using the proof of Theorem \ref{T:Doubling-general} given in Appendix \ref{A:closed formula}. 
For more general discussions regarding the realization of Clifford gates in the case of self-orthogonal codes with one or more more logical qubits, one can consult \cite{tansuwannont2025clifford}. 

\begin{figure*}[t]
\centering
\begin{tikzpicture}
\scalebox{0.8}{\begin{groupplot}[
    group style={
        group size=2 by 1,
        horizontal sep=5cm,
    },
    width=0.45\linewidth,
    height=0.35\linewidth,
    grid=major,
    domain=3:20,
    samples=300,
    xlabel={$d$},
    ylabel={$n$},
    scaled y ticks=false,
    y tick label style={/pgf/number format/fixed},
]

\nextgroupplot[
    title={2D Families},
    legend pos=north east,
    legend cell align=left,
    legend style={font=\scriptsize, fill=white, draw=none, at={(1.05,0.5)}, anchor=west},
]
\addplot[blue, thick, mark=*, mark repeat=10] {(x^2 + 2*x - 1)/2};
\addlegendentry{Our 2D}
\addplot[red, thick, dashed, mark=square*, mark repeat=10] {(3*(x^2-1)/4 + 1)};
\addlegendentry{Traditional 2D}
\addplot[green!60!black, thick, dotted, mark=triangle*, mark repeat=10] {(3*(x^2-1)/2 + 3)};
\addlegendentry{Capped 2D}

\nextgroupplot[
    title={3D Families},
    legend pos=north east,
    legend cell align=left,
    legend style={font=\scriptsize, fill=white, draw=none, at={(1.05,0.5)}, anchor=west},
]
\addplot[blue, thick, mark=*, mark repeat=10] {(x^3 + 6*x^2 + 5*x - 6)/6};
\addlegendentry{Our 3D}
\addplot[black, thick, dashed, mark=square*, mark repeat=10] {(x^3 + 5*x^2-x-9)/4};
\addlegendentry{Bravyi-Cross}
\addplot[red, thick, dashed, mark=square*, mark repeat=10] {(x^3 + x)/2};
\addlegendentry{Traditional 3D}
\addplot[green!60!black, thick, dotted, mark=triangle*, mark repeat=10] {(x^3 + 3*x^2 + 3*x -3)/4};
\addlegendentry{Recursive capped}
\addplot[purple, thick, dash pattern=on 3pt off 2pt, mark=diamond*, mark repeat=10] {(3*x^3 - 3*x^2 + x + 3)/4};
\addlegendentry{Stacked}
\end{groupplot}}
\end{tikzpicture}
\caption{Data qubit overhead $n$ (code length) as a function of code distance $d$ for 2D (left) and 3D (right) quantum families.}
\label{fig: comparison}
\end{figure*}
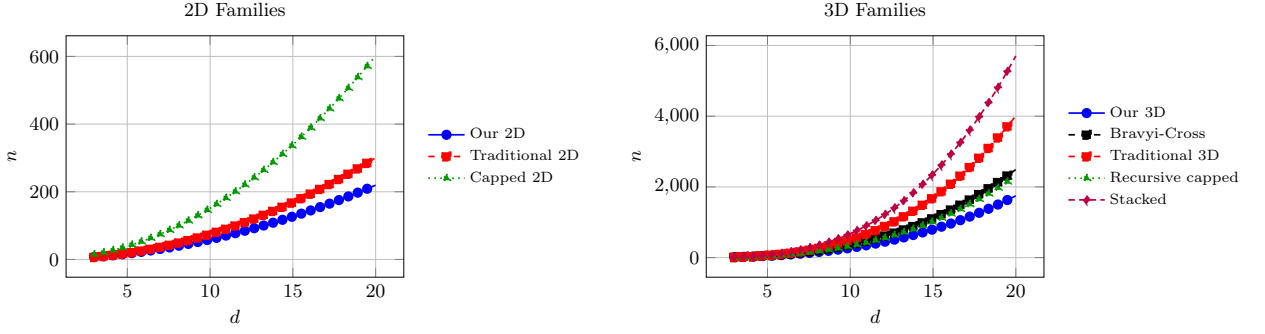

It is important to note that the codes constructed in Theorem \ref{T:Doubling color code1} can be substituted with other carefully chosen codes, and this can lead to other families of quantum codes that have transversal logical Clifford gates. 

%%%%%%%%%%%%%%%%%%%%%%%%%%%%%%%%%%%%%%%%%%%%%%
\subsection{Reaching finer logical Z-rotations}
%%%%%%%%%%%%%%%%%%%%%%%%%%%%%%%%%%%%%%%%%%%%%%

In this section, we extend the discussion presented in the previous section to build quantum codes with an arbitrary odd minimum distance that contain a transversal logical $Z$-rotation belonging to each Clifford hierarchy. 

Our proof is based on doubling the codes constructed from Theorem \ref{T:Doubling color code1} to 
first lift them to the Clifford level $\mathcal{C}^{(3)}$. Later we extend this result to each finer Z-rotation gate.

\begin{theorem}\label{T:CH3 family}
There exists a family of quantum (3D color) codes with parameters 
\[
[\![\frac{2}{3}k(k+1)(2k+1)-(2k+1),1,2k-1]\!]
\]
that realizes the logical $T$ gate transversally for each $k\ge 1$. Moreover, all $X$-stabilizers have weights divisible by $8$ (triply even).
\end{theorem}
The proof is provided in Appendix \ref{A:closed formula}.

Before proceeding to our next result, which gives codes with a logical gate from each level of the Clifford hierarchy, we establish some remarks and notation that are essential for later. 

We will provide a single code for each Clifford level and odd distance in our proposed family below. 
In this family, the length of the code that realizes the logical $Z$-rotation gates of the Clifford hierarchy $\mathcal{C}^{(r)}$ and has distance $2k-1$ will be denoted by $S_r(k)$. 
Moreover, as we apply the result of Theorem \ref{T:Doubling-general}, the length of codes within the family satisfies the identity  
\begin{equation}\label{equ:first form}
S_r(k)=2\big(\sum_{i=1}^{k}S_{r-1}(i) \big)-1
\end{equation}
with $S_r(1)=1$ (the $[\![1,1,1]\!]$ code implies this).
This is because the doubling construction requires 
\begin{itemize}
    \item codes of length $S_{r-1}(k)$ (two of them) 
    \item and one code with the length $S_r(k-1)$ 
\end{itemize}
to form a code of length $S_r(k)$.
Later in Appendix \ref{A:closed formula}, we prove that 
\begin{equation}\label{equ: recurrence length}
S_r(k)=\sum_{i=0}^{r} 2^i \binom{i+k-2}{i}
\end{equation}
and also that 
\begin{equation}\label{equ:divisibility}
S_r(k-1)+S_{r-1}(k)= 2^r \binom{k+r-2}{r}.
\end{equation}
Using these properties one can prove the following result. 
We discuss the proof in Appendix \ref{A:closed formula}. 

\begin{theorem}\label{T:general family form}
There exists a family of quantum codes with parameters 
\[
[\![S_r(k),1,2k-1]\!]
\]
that realizes logical $\frac{\pi}{2^{r-1}}$-$Z$ rotation gate through a transversal action of $\frac{\pi}{2^{r-1}}$-$Z$ rotation gates on the data qubits. \\
The $X$-stabilizers of a $[\![S_r(k),1,2k-1]\!]$ code have weights divisible by $2^r$, for each $k\ge 1$ and $r\ge 1$.
\end{theorem}

Table \ref{tab:parameters small} shows the parameters of the short length codes obtained from the above construction supporting a logical $Z$-rotation from $\mathcal{C}^{(2)}-\mathcal{C}^{(5)}$ of the Clifford hierarchy.

\begin{table*}[t]
    \centering  
    \renewcommand{\arraystretch}{1.4}
    \setlength{\tabcolsep}{6pt}
\[
\begin{array}{ l c c c c }
\hline
k \setminus r & 2 & 3 & 4 & 5 \\ \hline
2 & [\![7, 1, 3]\!] & [\![15, 1, 3]\!] & [\![31, 1, 3]\!] & [\![63, 1, 3]\!] \\ \hline
3 & [\![17, 1, 5]\!] & [\![49, 1, 5]\!] & [\![129, 1, 5]\!] & [\![321, 1, 5]\!] \\ \hline
4 & [\![31, 1, 7]\!] & [\![111, 1, 7]\!] & [\![351, 1, 7]\!] & [\![1023, 1, 7]\!] \\ \hline
5 & [\![49, 1, 9]\!] & [\![209, 1, 9]\!] & [\![769, 1, 9]\!] & [\![2561, 1, 9]\!] \\ \hline
\end{array}
\]
\caption{Parameters of small length codes from Theorem \ref{T:general family form} supporting logical $\frac{\pi}{2^{r-1}}$-$Z$ rotation. These codes are $2^r$ divisible.}
\label{tab:parameters small}
\end{table*}

Furthermore, a comparison between the code of Theorem \ref{T:general family form} for $r=2,3$ with some other known families of 2D and 3D color codes (with one logical qubit) has been provided in Figure \ref{fig: comparison}. 
We stress that, in this comparison, although our code families require a lower number of data qubits, they also have a few high weight stabilizer generators. 

%%%%%%%%%%%%%%%%%%%%%%%%%%%%%%%%%
\subsection{$Z$ Meta-checks in the doubled codes}
%%%%%%%%%%%%%%%%%%%%%%%%%%%%%%%%%%
In this section, we describe a few simple remarks about the structure of $Z$-checks in the quantum doubled codes that can enable single shot decoding of $Z$-syndromes.

As we mentioned in Section \ref{Sec:triorthogonal}, the $X$-stabilizers of a quantum code $Q$ obtained from applying the doubling construction has a matrix form 
\begin{equation*}
 \begin{bmatrix}
        E_1 & E_1&\bf{0_{n_2}}\\
        \bf{0_{n_1}}&\bf{0_{n_1}}&E_2\\
        \bf{0_{n_1}}&\bf{1_{n_1}}&\bf{1_{n_2}}
    \end{bmatrix},
\end{equation*}
where $E_1$ and $E_2$ are generators of the $X$-stabilizers of $Q_1$ and $Q_2$, respectively. 
Moreover, this code has a $X$-logical operator $({\bf{1_{n_1}}},{\bf{1_{n_1}}},{\bf{1_{n_2}}})$. 
A simple calculation shows that any $Z$-stabilizer of such a code has a binary representation as a linear combination of the following vectors:
\begin{enumerate}
    \item $(u_1,{\bf{0_{n_1}}},{\bf{0_{n_2}}})$, where $u_1$ forms a $Z$-stabilizer of $Q_1$. 
    \item $(w_1,w_1,{\bf{0_{n_2}}})$, where $w_1$ is an arbitrary even weight vector in $\F_2^{n_1}$,
    \item $({\bf{0_{n_1}}},{\bf{0_{n_1}}},u_2)$, where $u_2$ forms a $Z$-stabilizer of $Q_2$, and
    \item $({\bf{0_{n_1}}},{\bf{1_{n_1}}},{\bf{1_{n_2}}})$.
\end{enumerate}
Choosing $w_1$ belonging to a basis of the all-even code, i.e., 
\[
B=\{(1,1,0,\ldots,0), (0,1,1,\ldots,0),\ldots,(0,\ldots,0,1,1)\}
\]
produces a list of linearly independent vectors. 
The first takeaway of this representation is that all quantum codes with distance $d\ge5$ obtained from the doubling construction are degenerate because of the type $(2)$ vectors above. 
This includes, for example, the well-known 2D and 3D color codes of distance five, namely $[\![17,1,5]\!]$ and $[\![49,1,5]\!]$. 

Taking advantage of the above representations of $Z$-checks, one can apply the following simple tricks to detect and correct certain errors on the $Z$-syndromes. 

Let $C$ be an arbitrary binary linear code with parameters $[m,n_1-1,d]$ and with parity check matrix $H$. 
Then one can add $m-n_1+1$ new $Z$-checks (in total $m$ $Z$-checks instead of $n_1-1$), following the equations implied by $H$, to protect the type $(2)$ $Z$-syndromes. 
Indeed, one can
correct any $\lfloor \frac{d-1}{2} \rfloor$ (respectively detect any $d-1$) number of errors. 
The matrix $H$ is called a {\em meta-check} matrix for such syndromes.

Note that this approach is not necessary weight preserving as the new $Z$-checks can have larger weights. 
The following scenarios help to add new redundant $Z$-stabilizers of the same weight and protect the $Z$-syndromes. 
\begin{itemize}
    \item adding all elements of $B$ (defined above) gives a new $Z$-check of the form (2) that enables detection of each single error. 
    \item adding new $n_1-1$ redundant checks, obtained from the linear combination of each two type $(2)$ syndromes with consecutive elements of $B$, gives new weight four checks that can ``correct'' a single error on the type (2) checks. 
    \item for each  type (1) $Z$-check $a=(u_1,{\bf{0_{n_1}}},{\bf{0_{n_2}}})$  one can choose a new $Z$-check in the form of $b=(w,w,{\bf{0_{n_2}}})$ with $\wt(w)=2$ and support of $w$ subset of support of $u_1$. 
    Then the vector $a+b$ has the same weight as $a$ and adding enough $b$ and $a+b$ vectors helps to correct any single error on type (1) syndromes. 
    \item in many cases (such as $[\![17,1,5]\!]$ and $[\![49,1,5]\!]$) the code $Q_2$ itself is a doubled quantum code. Thus one can find redundant $Z$-checks (like the one mentioned above) for the type $(3)$ $Z$-syndromes. 
\end{itemize}

%%%%%%%%%%%%%%%%%%%%%%%%%%%%%%%%%%%%%%%%%%%%
\section{Rotated surface codes as another ingredient}\label{Sec: Surface as another}
%%%%%%%%%%%%%%%%%%%%%%%%%%%%%%%%%%%%%%%%%%%%

So far, we have focused exclusively on the doubling technique and color codes as its natural ingredient to construct larger quantum codes with improved minimum distance and desired logical $Z$-rotation gates.
The main purpose of this section is to emphasize ``how to move beyond color codes'' and incorporate other well-known families of codes into this framework.

Recall that the primary ingredient in the previous section was the all-even code, which enabled the construction of quantum codes admitting transversal Clifford gates. 
Recursively applying the same technique further allowed us to implement logical gates from higher levels of the Clifford hierarchy in a transversal manner.
Although the all-even  code is an optimal choice in terms of the number of data qubits, 
alternative candidates with other connectivity and geometry restrictions can be fed into the construction of Corollary
\ref{T:Doubling-selforthogonal}, or more generally Theorem \ref{T:Doubling-general}. 
In particular, one may apply the above construction to any odd-length quantum code whose 
$X$-stabilizers all have even weights to obtain the desired logical gates. 
Among these codes, the all-even code achieves the minimal length.

Before presenting our results, we highlight an important remark regarding the quantum codes discussed in this section.

\begin{remark}
The quantum codes constructed in this section possess strictly more $Z$-stabilizers than $X$-stabilizers. Equivalently, they are CSS codes specified by $C_2 \subseteq C_1$ with $C_1 \subsetneq C_2^\perp$. 
Hence, only the $X$-stabilizers constitute a self-orthogonal code. Accordingly, we refer to these codes as \emph{quantum $X$-self-orthogonal} codes, to distinguish them from fully self-orthogonal codes.
\end{remark}   

We use the following definition very frequently in our discussions. 

\begin{definition}\label{def:local geometr}
We say that a quantum code $Q$ has a local geometry of another quantum code $Q'$ in its geometry, if there exists a subset of data qubits $A$ such that:
\begin{itemize}
    \item \textbf{Puncturing rule}: restricting $X$-stabilizers and $X$-logical operators to $A$ (by discarding other parts) gives $X$-stabilizers and logical operator of $Q'$,
    \item \textbf{Shortening rule}: the set of $Z$-stabilizers and $Z$-logical operators with all supports on $A$ give $Z$-stabilizers and logical operator of $Q'$. 
\end{itemize}     
\end{definition}

We will restrict our attention to the quantum codes with a local geometry of rotated surface codes in their geometry.  

Recall that {\em rotated surface codes} are quantum CSS codes with parameters $[\![ d^2, 1, d ]\!]$. 
They have many interesting features including local geometry with nearest neighbor interactions, high circuit-level threshold, good decoders and well-understood logical operations \cite{demarti2024decoding,fowler2012surface}. Furthermore, these codes have been experimentally implemented in superconducting and neutral atom platforms \cite{Acharya2025,Bluvstein2024}. 
In a rotated surface code, physical qubits are positioned at the vertices of a $d \times d$ square grid for encoding one logical qubit. 
Its stabilizer generators consist of $X$-type and $Z$-type Pauli operators that are mapped to the faces (plaquettes) of the grid.  

Let $d \ge 3$ be a positive odd integer, and let $E$ denote the binary matrix whose rows correspond to the generators of the $X$-stabilizers of the rotated surface code of distance $d$. Then $E$ is a $\frac{d-1}{2} \times d^2$ binary matrix and its rows have even weights. Moreover, the vector $\mathbf{1}_{d^2}$ corresponds to the maximum-weight logical $X$ operator. 

A slight modification of Corollary~\ref{T:Doubling-selforthogonal}, via doubling the distance $d+2$ surface code and a distance $d$ self-orthogonal code, yields the following result. 
We omit the proof here.

\begin{proposition}\label{P:doubling-surface}
Let $Q$ be a binary self-orthogonal CSS code of odd length $n$, defined by $C_2 \subseteq C_1$, where $C_1 = C_2 \oplus \Span\{\mathbf{1}_n\}$, and with parameters $[\![n,1,d]\!]$, where $d$ is odd. 
Then one can construct a quantum $X$-self-orthogonal code $Q'$ with parameters
\[
[\![n + 2(d+2)^2, 1, d+2]\!],
\]
with a local geometry of rotated surface code.
\end{proposition}
We now highlight several important properties of the code $Q'$ appearing in the above proposition.
The $X$-stabilizers of $Q'$ are generated by
\begin{equation*}
G = \begin{bmatrix}
        E & E&\bf{0_{n}}\\
        \bf{0_{(d+2)^2}}&\bf{0_{(d+2)^2}}&G_2\\
        \bf{0_{(d+2)^2}}&\bf{1_{(d+2)^2}}&\bf{v}
    \end{bmatrix},
\end{equation*}
where $E$ and $G_2$ are the generator of $X$-stabilizers of rotated surface code with distance $d+2$ and the code $Q$, respectively, and $\bf{v}$ is a logical $X$-operator of $Q$ (of the minimum weight $d$). 
A logical $X$-operator of $Q'$, not necessarily of minimum weight, can be chosen as $(\bf{1_{(d+2)^2}},\bf{0_{(d+2)^2}},\bf{0_n})$. 
Moreover, if $\bf{u}$ is a logical $Z$-operator of the distance $d+2$ surface code, then the code $Q'$ has a logical $Z$-operator of the form $(\bf{u},\bf{0_{(d+2)^2}},\bf{0_n})$ (of minimum weight $d+2$).

The main feature of this construction is that it allows us to preserve the local geometry of the rotated surface code within a larger quantum code that has the transversal implementation of the logical $S$ gate. 
As we will see later, such a local geometry is essential in our FT code switching protocol for realizing logical $S$, $T$, or other finer $Z$-rotation gates.

Next, we recursively apply the result of Proposition~\ref{P:doubling-surface} to construct an infinite family of $X$-self-orthogonal codes based on rotated surface codes. 
We summarize the parameters of these codes for any desired odd minimum distance in the following theorem.

\begin{theorem}\label{T:doubling-surface}
There exists a family of 3D quantum $X$-self-orthogonal codes with the local geometry of rotated surface codes, in which the logical $S$ gate is realized transversally, and with parameters
\[
[\![ \frac{d(d+1)(d+2)}{3}-1 , 1, d]\!],
\]
for odd $d \ge 1$. \\
Moreover, when $d \ge 3$, the set of all $X$-stabilizers can be decomposed as $C = C_e \oplus \Span(\{v\})$, where $C_e$ contains vectors of weights divisible by four and $ \wt(v)\equiv 2 \pmod 4$.
\end{theorem}

The proof is given in Appendix \ref{A:supp-surface}.
The 3D structure and the recursive construction of such codes is depicted in Figure \ref{fig: Surface}.
In particular, when $d=3$ and $d=5$, one obtains quantum $X$-self-orthogonal codes with parameters $[\![19,1,3]\!]$ and $[\![69,1,5]\!]$, respectively, that enable the transversal realization of the logical $S$ gate (see Figure \ref{fig: Surface}). 
We believe that these codes are highly likely to support logical $T$ gate realizable in a fold-transversal manner. 
However, we postpone the careful study of such gates to future research.

We would also like to remark that the codes constructed from Theorem \ref{T:doubling-surface} have a symmetric geometry consisting of different layers of rotated surface codes, starting from the distance $1$ surface code ($[\![1,1,1]\!]$ code) at the bottom and ending with the distance $d$ surface code at the top. 
Moreover, two representatives of minimum weight logical $X$- and $Z$-stabilizers can be specified by data qubits of the top layer surface code and its horizontal edge, respectively (see Figure \ref{fig: Surface}). 

\begin{figure}[h]
\centering
\begin{minipage}{0.20\textwidth}
\centering
\scalebox{0.60}{
\begin{tikzpicture}[
    x={(1cm,0cm)}, 
    y={(0.5cm,0.5cm)}, 
    z={(0cm,1cm)},
    plane/.style={fill opacity=0.2, thick}
]
  \coordinate (L1) at (0,0,0); \coordinate (L2) at (3,0,0);
  \coordinate (L3) at (3,3,0); \coordinate (L4) at (0,3,0);
  \coordinate (LC) at (1.5,1.5,0); % Center
  \coordinate (U1) at (0,0,2.5); \coordinate (U2) at (3,0,2.5);
  \coordinate (U3) at (3,3,2.5); \coordinate (U4) at (0,3,2.5);
  \coordinate (A) at (1.5, 1.5, -2.5);
\coordinate (AA) at (1.5, 1.5, -3.5);
  \draw[dotted, thick, black!30] (A) -- (L3);
  \draw[dotted, thick, black!30] (A) -- (L4);

  \filldraw[plane, fill=blue!30, draw=blue!70] (L1) -- (L2) -- (L3) -- (L4) -- cycle;
  \node[blue!80!black, anchor=east] at (L1) {$d=3$};
  \draw[thick, fill=blue!15] (1,1,0) -- (2,1,0) -- (2,2,0) -- (1,2,0) -- cycle;

  \draw[dotted, thick, black!80] (A) -- (L1);
  \draw[dotted, thick, black!80] (A) -- (L2);
  \draw[dotted, thick, black!80] (A) -- (LC);
  
  \fill[purple!90] (A) circle (3pt) node[right=5pt, font=\small\bfseries] {};

  \fill[blue] (AA) circle (3pt) node[right=5pt, font=\small\bfseries] {Ancilla qubit};

  \draw[dotted, thick, black!80] (A) -- (AA);

  \filldraw[plane, fill=blue!25, draw=blue!70] (U1) -- (U2) -- (U3) -- (U4) -- cycle;
  \node[red!80!black, anchor=east] at (U1) {$d=3$};
  \draw[red, ultra thick] (U3) -- (U4); 
  \draw[blue, ultra thick] (U4) -- (U1);
   \draw[blue, ultra thick] (L4) -- (L1);
  \draw[thick, fill=blue!15] (1,1,2.5) -- (2,1,2.5) -- (2,2,2.5) -- (1,2,2.5) -- cycle;
  \coordinate (M) at (1.5, 1.5, 1.25);
  \fill[purple!90] (M) circle (2.5pt) node[right=4pt, font=\small\bfseries] {Local Check};
  
  \foreach \i in {1,2}{
    \draw[dashed, black!30] (M) -- (1,\i,0); \draw[dashed, black!30] (M) -- (2,\i,0);
    \draw[dashed, black!30] (M) -- (1,\i,2.5); \draw[dashed, black!30] (M) -- (2,\i,2.5);
  }

\end{tikzpicture}
}
\end{minipage}
\hfill
\begin{minipage}{0.20\textwidth}
{\centering
\hspace{-1.5cm}\scalebox{0.60}{
\begin{tikzpicture}[
    x={(1cm,0cm)}, 
    y={(0.5cm,0.5cm)}, 
    z={(0cm,1.5cm)}, 
    plane/.style={fill opacity=0.15, thick}
]

  \filldraw[plane, fill=blue!30, draw=blue!70] (0,0,3.0) -- (3,0,3.0) -- (3,3,3.0) -- (0,3,3.0) -- cycle;
  \node[blue!80!black, anchor=east] at (0,0,3.0) {\small $d=5$};
  \draw[thick, fill=blue!10] (1,1,3.0) -- (2,1,3.0) -- (2,2,3.0) -- (1,2,3.0) -- cycle;
  \draw[blue, ultra thick] (0,0,3.0) -- (0,3,3.0) ;

  \filldraw[plane, fill=blue!25, draw=blue!70] (0,0,4.2) -- (3,0,4.2) -- (3,3,4.2) -- (0,3,4.2) -- cycle;
  \node[red!80!black, anchor=east] at (0,0,4.2) {\small $d=5$};
  \draw[thick, fill=blue!10] (1,1,4.2) -- (2,1,4.2) -- (2,2,4.2) -- (1,2,4.2) -- cycle;
  \draw[red, ultra thick] (3,3,4.2) -- (0,3,4.2);
  \draw[blue, ultra thick] (0,0,4.2) -- (0,3,4.2);
  \coordinate (M1) at (1.5, 1.5, 3.6);
  \fill[purple!90] (M1) circle (2.2pt) node[right=3pt, font=\bfseries\tiny] {Local Check};
  \foreach \z in {3.0, 4.2} \foreach \x/\y in {1/1, 2/1, 2/2, 1/2} \draw[dashed, black!40, thin] (M1) -- (\x,\y,\z);

  \coordinate (A1) at (1.5, 1.5, 2.0); 
  
  \foreach \p in {(0,0,3.0), (3,0,3.0), (3,3,3.0), (0,3,3.0), (1.5,1.5,3.0)}
    \draw[dotted, thick, black!50] (A1) -- \p;
      \foreach \p in {(0,0,1.0), (3,0,1.0), (3,3,1.0), (0,3,1.0), (1.5,1.5,1.0)}
    \draw[dotted, thick, black!50] (A1) -- \p;

  \fill[purple!90] (A1) circle (3pt) node[right=5pt, font=\small\bfseries] {};
  
  \filldraw[plane, fill=blue!30, draw=blue!70] (0,0,1.0) -- (3,0,1.0) -- (3,3,1.0) -- (0,3,1.0) -- cycle;
  \node[blue!80!black, anchor=east] at (0,0,1.0) {\small $d=3$};
  \draw[thick, fill=blue!10] (1,1,1.0) -- (2,1,1.0) -- (2,2,1.0) -- (1,2,1.0) -- cycle;
    \draw[blue, ultra thick] (0,0,1.0)-- (0,3,1.0);

  \filldraw[plane, fill=blue!30, draw=blue!70] (0,0,-0.2) -- (3,0,-0.2) -- (3,3,-0.2) -- (0,3,-0.2) -- cycle;
  \node[blue!80!black, anchor=east] at (0,0,-0.2) {\small $d=3$};
  \draw[thick, fill=blue!10] (1,1,-0.2) -- (2,1,-0.2) -- (2,2,-0.2) -- (1,2,-0.2) -- cycle;
 \draw[blue, ultra thick] (0,0,-0.2)-- (0,3,-0.2) ;

  \coordinate (M2) at (1.5, 1.5, 0.4);
  \fill[purple!90] (M2) circle (2.2pt) node[right=3pt, font=\bfseries\tiny] {Local Check};
  \foreach \z in {1.0, -0.2} \foreach \x/\y in {1/1, 2/1, 2/2, 1/2} \draw[dashed, black!40, thin] (M2) -- (\x,\y,\z);

  \coordinate (A2) at (1.5, 1.5, -1.4);
  \coordinate (AA) at (1.5, 1.5, -2.0);
  \foreach \p in {(0,0,-0.2), (3,0,-0.2), (3,3,-0.2), (0,3,-0.2), (1.5,1.5,-0.2)}
    \draw[dotted, thick, black!50] (A2) -- \p;

  \fill[purple!90] (A2) circle (3pt) node[right=5pt, font=\small\bfseries] {};
 \fill[blue] (AA) circle (3pt) node[right=5pt, font=\small\bfseries] {Ancilla qubit};

  \draw[dotted, thick, black!80] (A2) -- (AA);
\end{tikzpicture}
}}
\end{minipage}
\caption{Example of codes constructed from Theorem \ref{T:doubling-surface}.
(Left) The geometry of $[\![19,1,3]\!]$ $X$-self-orthogonal code. 
The two planes represent two surface codes of distance 3, and the red dots show their X checks, where local checks are as many as the number of X-check of the surface codes.   
Such local checks are connected to data qubits of the two surrounding copies of surface codes. 
The lowest check is connected to all the data qubits of its top surface code, and the single data qubit of totally orthogonal code (named Ancilla qubit). 
The tick blue and red qubits (on the planes or the bottom qubit) specify the location of data qubits corresponding to minimum weight logical $X$ and $Z$ operators of the code, respectively. (Right) Recursive construction of $[\![69,1,5]\!]$.}
\label{fig: Surface}
\end{figure}
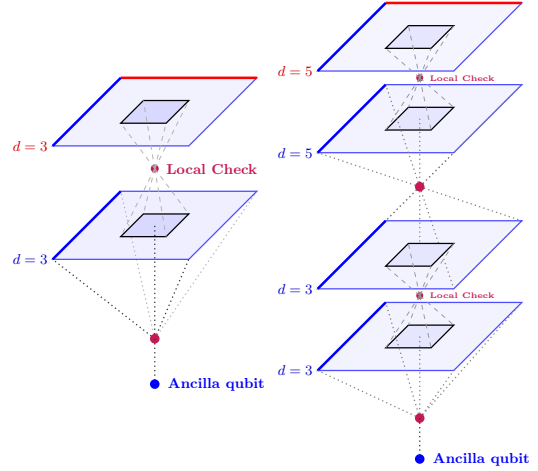

Note also that one can optimize the qubit overhead using Proposition \ref{P:doubling-surface}.  
For instance, applying this proposition to $[\![7,1,3]\!]$ Steane code and two copies of $d=5$ rotated surface code results in $[\![57,1,5]\!]$ quantum $X$-self-orthogonal code 
(while, Theorem \ref{T:doubling-surface} gives the $[\![69,1,5]\!]$ code shown in Figure \ref{fig: Surface}). 
We omit such optimization discussions in this section.

This process can be continued to construct quantum codes that are $r$-orthogonal and possess the local geometry of rotated surface codes, for each $r,k \ge 2$. 
Let 
\begin{equation}\label{equ:surface length}
S_r(k) = 2^{r+2} \binom{k+r-1}{r+1} + \sum_{i=0}^{r-1} 2^i \binom{k+i-2}{i}.
\end{equation}

\begin{theorem}\label{T:surface-hierarchy}
There exists a family of quantum codes with parameters 
\[
[\![S_r(k),1,2k-1]\!]
\]
that are $r$-orthogonal and with the local geometry of rotated surface codes of minimum distance $2k-1$. 
\end{theorem}
See Appendix \ref{A:supp-surface} for a proof. 
In particular, when $r=3$ one gets triorthogonal codes than can realize logical $T$ through the transversal application of physical $T$ gates, up to Clifford correction. 
Indeed, fixing $r=3$ simplifies \eqref{equ:surface length} to
\[
S_3(k) = \frac{4 k^4 + 8 k^3 + 2 k^2 - 8 k - 3}{3}.
\]

\begin{table*}[t]
\centering
\small
\renewcommand{\arraystretch}{1.4}
\setlength{\tabcolsep}{6pt}
\begin{tabular}{l c}
\hline
Family & Parameters $[\![n,k,d]\!]$ \\
\hline

$X$-self-orthogonal codes &
$[\![\frac{(2k-1)(2k)(2k+1)}{3}-1,\,1,\,2k-1]\!]$ \\

Triorthogonal codes &
$[\![\frac{4k^4+8k^3+2k^2-8k-3}{3},\,1,\,2k-1]\!]$ \\

$r$-orthogonal codes $r,k\ge 2$ &
$[\![2^{r+2}\binom{k+r-1}{r+1}
+\sum_{i=0}^{r-1}2^i\binom{k+i-2}{i},\,1,\,2k-1]\!]$ \\

\hline
\end{tabular}
\caption{Parameters of families of quantum codes with the local geometry of rotated surface codes discussed in this section.}
\label{tab:surface}
\end{table*}

\begin{corollary}\label{cor: tri-surface}
There exists a family of quantum triorthogonal codes with parameters
\[
[\![\frac{4 k^4 + 8 k^3 + 2 k^2 - 8 k - 3}{3},1,2k-1]\!]
\]
and with the local geometry of rotated surface codes.
\end{corollary}
For instance, when $k=2$, one gets 
a triorthogonal code with parameters $[\![39,1,3]\!]$ (see Figure \ref{fig:4D surface}). Table \ref{tab:surface} summarizes the codes constructed in this section.

\begin{figure}[h!]
{ \centering
\scalebox{0.6}{
\begin{tikzpicture}[
    x={(1cm,0cm)}, 
    y={(0.4cm,0.2cm)}, 
    z={(0cm,1.2cm)},
    plane/.style={fill opacity=0.1, thick}
]
  \def\xone{0.6}
  \def\xtwo{2.2}
  \def\xanc{-1.8}
  \def\xmid{1.4}
  \def\zmerged{2.5}

  \coordinate (A_merged) at (\xanc, 0, \zmerged);
  \fill[purple!90] (A_merged) circle (3.5pt);

  \coordinate (A11) at (\xanc - 1, 0, 4.0);
  \fill[black] (A11) circle (3.5pt);
  \draw[dotted, black!40] (A_merged) -- (A11);

  \coordinate (A21) at (\xanc - 1, 0, 1.0);
  \fill[black] (A21) circle (3.5pt);
  \draw[dotted, black!40] (A_merged) -- (A21);

  \foreach \x in {\xone, \xtwo}{
    \filldraw[plane, fill=gray!20, draw=black!70] (\x,-1,3.0) -- (\x,1,3.0) -- (\x,1,5.0) -- (\x,-1,5.0) -- cycle;
    \filldraw[purple!40, opacity=0.7, draw=purple, thick] (\x,-0.4,3.6) -- (\x,0.4,3.6) -- (\x,0.4,4.4) -- (\x,-0.4,4.4) -- cycle;
  }
  
  \foreach \y/\z in {-1/3.0, 1/3.0, 1/5.0, -1/5.0} 
    \draw[dashed, black!40] (A_merged) -- (\xone,\y,\z);

  \coordinate (M_bridge) at (\xmid, 0, 2.5);
  \fill[purple!90] (M_bridge) circle (3.5pt);

  \draw[dashed, purple!80, thick] (M_bridge) -- (\xone, 0, 3.6); 
  \draw[dashed, purple!80, thick] (M_bridge) -- (\xtwo, 0, 3.6);
  \draw[dashed, purple!80, thick] (M_bridge) -- (\xone, 0, 1.4);
  \draw[dashed, purple!80, thick] (M_bridge) -- (\xtwo, 0, 1.4);

  \foreach \x in {\xone, \xtwo}{
    \filldraw[plane, fill=gray!20, draw=black!70] (\x,-1,0) -- (\x,1,0) -- (\x,1,2.0) -- (\x,-1,2.0) -- cycle;
    \filldraw[purple!40, opacity=0.7, draw=purple, thick] (\x,-0.4,0.6) -- (\x,0.4,0.6) -- (\x,0.4,1.4) -- (\x,-0.4,1.4) -- cycle;
  }
  
  \foreach \y/\z in {-1/0, 1/0, 1/2.0, -1/2.0} 
    \draw[dashed, black!40] (A_merged) -- (\xone,\y,\z);

  \coordinate (M_final) at (\xmid, 0, -1.2);
  \fill[purple!90] (M_final) circle (3.5pt);
  
  \foreach \x in {\xone, \xtwo} {
      \foreach \y/\z in {-1/0, 1/0, 1/2.0, -1/2.0} {
          \draw[dashed, purple!40, thick] (M_final) -- (\x, \y, \z);
      }
  }

  \coordinate (Q) at (\xmid, 0, -1.8);
  \fill[black] (Q) circle (5pt);
  \draw[dashed, purple!90, thick] (M_final) -- (Q);
  \draw[dashed, purple!40, thick] (M_final) -- (A21);
\end{tikzpicture}
}}
\caption{Geometric representation of $[\![39,1,3]\!]$ triorthogonal code.}
\label{fig:4D surface}
\end{figure}
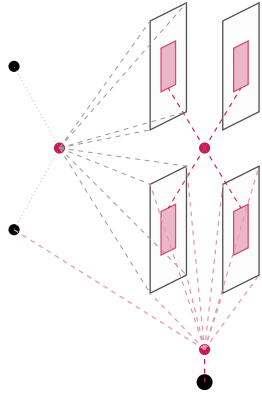

%%%%%%%%%%%%%%%%%%%%%%%%%%%%%%%%%%%
\subsection{A Qubit Overhead Optimization Technique}
%%%%%%%%%%%%%%%%%%%%%%%%%%%%%%%%%%%

In this section, we discuss a technique to reduce the qubit overhead of the codes constructed in the previous section, while preserving the local geometry of the rotated surface code in their representation. 
In certain cases, this also reduces the required qubit connectivity.

Recall that we used two copies of rotated surface codes in Theorem \ref{T:doubling-surface} to obtain families of quantum $X$-self-orthogonal codes with the local geometry of rotated surface codes, and then we extended this approach to $r$-orthogonal codes for each $r\ge3$. 
In this section, we propose a different approach that combines one rotated surface code with a punctured rotated surface code instead of using two copies of the surface codes. 
 
We begin with a preliminary example that illustrates the key ideas, and subsequently present the main result. 
We construct a quantum $X$-self-orthogonal code with parameters $[\![15,1,3]\!]$, obtained by horizontally concatenating
\begin{itemize}
    \item 
the $X$-check matrix of a distance-three rotated surface code, and 
\item that of an all-even code of length five (which is a punctured surface code).
\end{itemize}
The distance three rotated surface code has parameters $[\![9,1,3]\!]$ consisting of $\frac{9-1}{2}$ X-check generators $v_1$, $v_2$, $v_3$ and $v_4$ that form an even binary code space. 
One can easily verify, by inspecting Figure \ref{fig:surface-even} and labeling the bottom-to-top X-checks as $v_1$–$v_4$, that the non-zero inner products of these basis vectors occur only in the following cases:
\begin{itemize}
    \item $v_1 \cdot v_2=v_2 \cdot v_3=v_3 \cdot v_4=1$.
\end{itemize}
On the other hand, the all-even code of length five has basis vectors $u_i$ each having non-zero coordinates in positions $i$ and $i+1$, for each $1\le i \le 4$, and the non-zero inner products of the basis vectors occur only in the following cases: 
\begin{itemize}
    \item $u_1 \cdot u_2=u_2 \cdot u_3=u_3 \cdot u_4=1$.
\end{itemize}
Now one can define a new X-check matrix of size $9+5+1$ in the form: 
\begin{equation*}
 G=\begin{bmatrix}
        E_1 & E_2&0\\
        \bf{0_{9}}&\bf{0_{5}}&0\\
        \bf{0_{9}}&\bf{1_{5}}&1
    \end{bmatrix},
\end{equation*}
where $E_1$ and $E_2$ have rows $v_1-v_4$ and $u_1-u_4$, respectively. 
Using the above inner products, one can show that the rows of $G$ are mutually orthogonal and form a self-orthogonal binary linear code $C_2$. 
Define
\[
C_1 = C_2 \oplus \Span\{(\mathbf{1}_9,\mathbf{0}_5,0)\}.
\]
Then the CSS code associated with $C_2 \subset C_1$ is an $X$-self-orthogonal code with parameters $[\![15,1,3]\!]$. 
Moreover, linear combinations of the X-checks restricted to the first nine columns yield the X-checks of the $9$-qubit rotated surface code. 

This process is shown in Figure \ref{fig:surface-even} 
by counting each red data as two qubits, one for the surface code and one for the all-even code. 
The last step of the doubling construction is to add a new data qubit that is connected to all five qubits of all-even code through a new $X$ check. 
The $X$-stabilizers of the $[\![15,1,3]\!]$ self-orthogonal code along with some other interesting features of them are presented in Appendix \ref{A:15-1-3 stabilizers}. 

Applying another round of doubling, as in \eqref{E:general doubling matrix}, to two copies of the $[\![15,1,3]\!]$code and the $[\![1,1,1]\!]$ totally orthogonal code produces a $[\![31,1,3]\!]$ triorthogonal code with degenerate $X$ and $Z$ distances. 
Such a code contains a local geometry of the rotated surface code in its geometry.

\begin{figure}[h!]
    \centering
\includegraphics[width=0.9\linewidth]{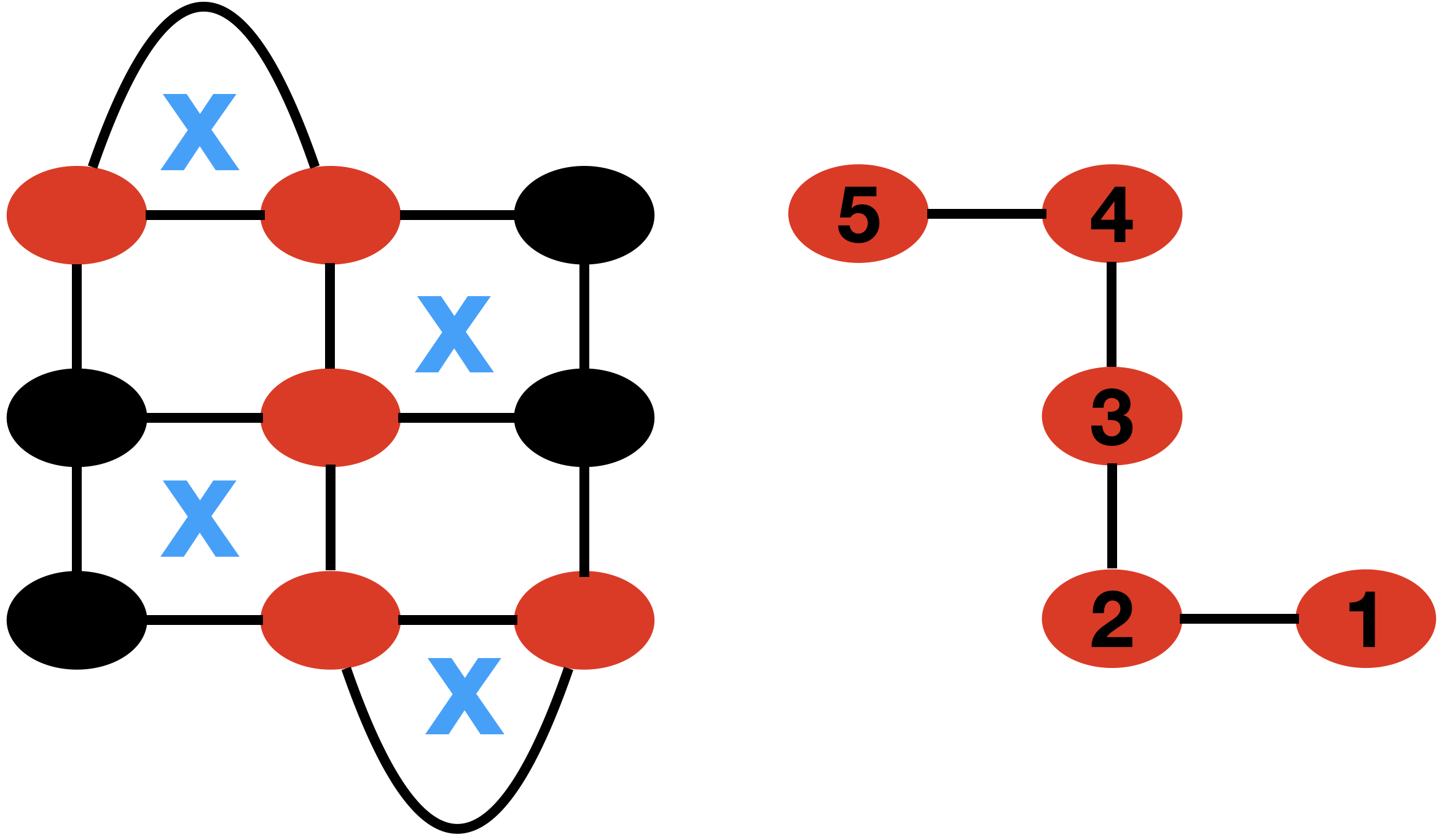}
    \caption{Tanner graph of all-even  code (right). A 14 qubit $X$-self-orthogonal code, with the same $X$ checks as surface codes, but counting each red dot as two data qubits (one for surface and one for all-even).
    Adding an additional data qubit and connecting it through an additional $X$ check to all five qubits of the all-even  code implies the quantum $X$-self-orthogonal $[\![15,1,3]\!]$.}
    \label{fig:surface-even}
\end{figure}

Next we generalize the idea for any odd distance rotated surface code. 
\begin{theorem}\label{T:Surface optimization}
Let $Q_1$ be a rotated surface code with parameters $[\![d^2,1,d]\!]$ for an odd $d \ge 3$ and $X$-stabilizer generator $E$. 
\begin{enumerate}
\item Then there exists a self-orthogonal binary linear code of length $\frac{3d^2+1}{2}$ with the generator matrix $[E \ E']$, where $E'$ is obtained after puncturing of certain columns of $E$.
\item If $Q_2$ is an $X$-self-orthogonal code with parameters $[\![n,1,d-2]\!]$ and $X$-stabilizer generator $G$, then there exists an $X$-self-orthogonal quantum code $Q$ with parameters 
\[
[\![n+2(d^2-d+1),1,d]\!]
\]
containing a local geometry of $Q_1$ in its geometry.
\end{enumerate}
\end{theorem}
The proof is given in Appendix \ref{A:symplectic}. 
Notice that the claimed quantum code $Q$ has $X$-check generator
\begin{equation*}
 G=\begin{bmatrix}
        E & E''&\bf{0_{n}}\\
        \bf{0_{d^2}}&\bf{0_{d^2-2d+2}}&G\\
        \bf{0_{d^2}}&\bf{1_{d^2-2d+2}}&\bf{1_{n}}
    \end{bmatrix},
\end{equation*}
where $E''$ is obtained from puncturing the $2d-2$ columns of matrix $E$ ($X$-generator of rotated surface code). Such removed columns correspond to data qubits that 
\begin{itemize}
    \item appear only once in all weight four $X$-checks, 
    \item and are not part of a weight two $X$-check (see Figure \ref{fig:puncture surface}).
\end{itemize} 

\begin{figure}
    \centering
\begin{tikzpicture}[scale=1.2]
    \definecolor{xcolor}{RGB}{255, 182, 193}
    \definecolor{xborder}{RGB}{200, 50, 50}

    \fill[xcolor] (1,4) -- (2,4) -- (1.5,4.5) -- cycle;
    \fill[xcolor] (3,4) -- (4,4) -- (3.5,4.5) -- cycle;

    \fill[xcolor] (0,0) -- (1,0) -- (0.5,-0.5) -- cycle;
    \fill[xcolor] (2,0) -- (3,0) -- (2.5,-0.5) -- cycle;

    \fill[xcolor] (0,3) -- (1,3) -- (1,4) -- (0,4) -- cycle;
    \fill[xcolor] (2,3) -- (3,3) -- (3,4) -- (2,4) -- cycle;
    
    \fill[xcolor] (1,2) -- (2,2) -- (2,3) -- (1,3) -- cycle;
    \fill[xcolor] (3,2) -- (4,2) -- (4,3) -- (3,3) -- cycle;
    
    \fill[xcolor] (0,1) -- (1,1) -- (1,2) -- (0,2) -- cycle;
    \fill[xcolor] (2,1) -- (3,1) -- (3,2) -- (2,2) -- cycle;
    
    \fill[xcolor] (1,0) -- (2,0) -- (2,1) -- (1,1) -- cycle;
    \fill[xcolor] (3,0) -- (4,0) -- (4,1) -- (3,1) -- cycle;

    \draw[gray!30, very thin] (0,0) grid (4,4);
    \foreach \x in {0,1,2,3,4} {
        \foreach \y in {0,1,2,3,4} {
            \node[draw, circle, fill=white, inner sep=1.8pt, minimum size=6pt] at (\x,\y) {};
        }
    }
    \foreach \x/\y in {
        4/4,1/4,
        1/3,1/2,1/1,1/0,
        2/0,
        2/1,2/2,2/3,2/4,
        3/4,
        3/3,3/2,3/1,3/0,
        0/0
    } {
        \node[draw, circle, fill=black, inner sep=1.8pt, minimum size=6pt] at (\x,\y) {};
    }

\end{tikzpicture}
\caption{Geometry of punctured surface codes. 
    The data qubits of the new code are represented by the black dots and these black dots share the same X-check structure as the rotated surface code.  
    For any two chosen red stabilizers, the inner product of their corresponding binary vectors is the same in both the punctured code and the rotated surface code.}
    \label{fig:puncture surface}
\end{figure}
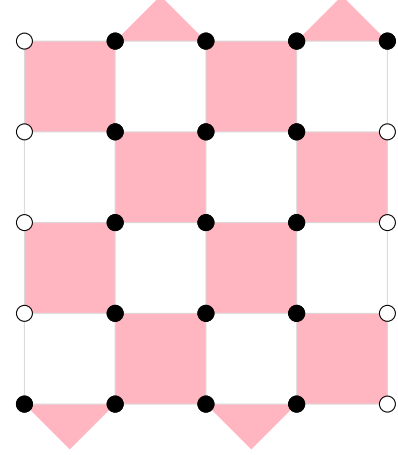

Note also that replacing the block $[E \ E'']$ with $[E \ E']$ from part (1) may yield a different quantum $X$-self-orthogonal code of smaller size; however, the corresponding claim about the distance does not necessarily hold.

The codes obtained from Theorem \ref{T:Surface optimization} (2) will be called {\em DSPS (doubled surface and punctured surface)} codes. 
The distance three DSPS, as discussed above, has parameters $[\![15,1,3]\!]$. 

\begin{example}
Here we construct quantum $X$-self-orthogonal and triorthogonal codes of distance five using rotated surface codes. The codes presented in this section require less number of data qubits compared to the codes presented in the previous section.

First recall that, as we mentioned in the previous section, there exist $X$-self-orthogonal codes with parameters $[\![69,1,5]\!]$ and $[\![57,1,5]\!]$ and local geometry of rotated surface codes. 
Also Corollary \ref{cor: tri-surface} implies the existence of a $[\![177,1,5]\!]$ triorthogonal code with a local surface code geometry. 

First we apply the result of Theorem \ref{T:Surface optimization} (2) to the $[\![25,1,5]\!]$ rotated surface code and the $[\![7,1,3]\!]$ Steane code. 
This gives a $[\![49,1,5]\!]$ quantum $X$-self-orthogonal code with a local geometry of rotated surface codes. 

Now applying the Quantum Doubling construction to the latter code and $[\![15,1,3]\!]$ implies a $[\![113,1,5]\!]$ quantum triorthogonal code with a local geometry of rotated surface codes (which requires much less qubits than the $[\![177,1,5]\!]$ code). 
\end{example}

Using a similar process, one can obtain the list of DSPS and triorthogonal codes presented in Table \ref{tab: DSPS parameters small}. 

\begin{table*}[t]
\centering 
\renewcommand{\arraystretch}{1.4}
\setlength{\tabcolsep}{6pt}
\[
\begin{array}{l c c c c}
\hline
\text{Surface code}& [\![9, 1, 3]\!] & [\![25, 1, 5]\!] & [\![49, 1, 7]\!] & [\![81, 1, 9]\!] \\ \hline
\text{DSPS ($X$-self-orthogonal) code}& [\![15, 1, 3]\!] & [\![49, 1, 5]\!] & [\![103, 1, 7]\!] & [\![169, 1, 9]\!] \\ \hline
\text{Triorthogonal (doubled DSPS) code}& [\![31, 1, 3]\!] & [\![113, 1, 5]\!] & [\![255, 1, 7]\!] & [\![433, 1, 9]\!] \\ \hline
\end{array}
\]
\caption{Parameters of small length DSPS and triorthogonal quantum codes with a local geometry of surface codes.}  
\label{tab: DSPS parameters small}
\end{table*}

%%%%%%%%%%%%%%%%%%%%%%%%%%%%%%%%%%%%%
\section{General Transversal Code Switching}\label{Sec: code switching}
%%%%%%%%%%%%%%%%%%%%%%%%%%%%%%%%%%%%%%

Code switching originated as a gauge-fixing technique designed to alternate between two distinct topological color codes sharing a common underlying subsystem code \cite{bombin2015gauge,kubica2015universal}. 
By dynamically altering which gauge operators are measured, 
the system can be projected into either a 2D or 3D color code configuration. 
This framework provides a path to universal fault-tolerant quantum computing without magic state distillation. 
A non-Clifford gate can be executed transversally while the system is fixed in the 3D configuration, after which a gauge switch safely returns the state to the 2D configuration for standard Clifford processing and lower-overhead error correction.  
The protocol was subsequently refined to achieve improved fidelity and has since been validated through both numerical simulations and experimental verification \cite{heussen2025efficient,daguerre2025code,daguerre2025experimental} for the fault-tolerant implementation of the $T$ gate. 
The main scheme is based on a qubit $Z$-teleportation considered in \cite{zhou2000methodology}.   

In this section, we restrict our attention to the code switching protocol based on the transversal CNOT between the two complementary sets of gates. 
In particular, 
we generalize the protocol for the fault-tolerant implementation of an arbitrary fine $Z$-rotation gate, and encompass all quantum codes described in the previous sections into this protocol.
A notable feature of the proposed scheme is that it is not restricted to any specific code family.
More precisely, the main requirement for such code switching protocol is a quantum code that has the geometry of another quantum code in its overall geometry (see Definition \ref{def:local geometr}). 
In particular, we have the following result.

\begin{theorem}\label{T: Magic teleportation}
Let $Q_1$ and $Q_2$ be quantum codes of lengths $n_1$ and $n_2$, with $n_2<n_1$, each encoding a single logical qubit such that $Q_1$ has a local geometry of $Q_2$ in its first $n_2$ positions. 
Then 
\begin{enumerate}
    \item there exists a transversal CNOT gate between $Q_1$ and $Q_2$ in the form 
    \[
\overline{\CNOT}=\bigotimes_{i=1}^{n_2}\CNOT_i.
    \]
    \item If $Q_1$ realizes the logical $R_Z(\theta)$ gate through the physical operation $O$, then the outcome of the following circuit is the $\ket{R_Z(\theta)}$ on $Q_2$:
    \[
\scalebox{0.8}{
\quad \quad\quad\quad\quad\quad\Qcircuit @C=0.001em @R=.1em  @!{\lstick{Q_1\:\:\:\:|{+}\rangle} &\gate{O}   &\ctrl{1}   & \measuretab{M_{X_L}} & \control \cw & & & \\
\lstick{Q_2\:\:\:\:|{0}\rangle} &\qw  &\targ & \qw & \gate{Z_L} \cwx & \qw &&\lstick{\ket{R_Z(\theta)}}
}}
\]
\end{enumerate}
\end{theorem}

\begin{proof}
(1) Let $Q_1$ and $Q_2$ be two quantum codes with the given conditions. We have 
\begin{enumerate}
    \item for each $X$-stabilizer (or logical operator) $O=\otimes_{i=1}^{n_1}O_i$ of $Q_1$, the operator $O_{11}=\otimes_{i=1}^{n_2}O_i$ is an $X$-stabilizer (or logical operator) of $Q_2$.
    \item for each $Z$ stabilizer or logical operator $O=\otimes_{i=1}^{n_2}O_i$ of $Q_2$,  $O$ is a $Z$-stabilizer (or logical operator) of $Q_1$.

 \item The $X$- and $Z$-logical operators of both $Q_1$ and $Q_2$ are based on applying transversal $X$ and $Z$ operators acting on  the first $n_2$ qubits, respectively.
\end{enumerate}
Now, consider the direct sum of the $Q_1$ and $Q_2$ codes, which is a new quantum (CSS) code $S$ with two logical qubits stabilized by the stabilizers of $Q_1$ and $Q_2$. 
The set of all stabilizers of $S$ are in the form of $O_1\otimes O_2$ where $O_1$ is a stabilizer of $Q_1$ and $O_2$ is a stabilizer of $Q_2$.

The transversal CNOT gate is defined by
\begin{equation}\label{equ: cnot gate}
\overline{\CNOT}=\bigotimes_{i=1}^{n_2}\CNOT_i,
\end{equation}
where the gate $\CNOT_i$ has its control on the $i$-th qubit of $Q_1$ and its target is on the $i$-th qubit of $Q_2$, for each $1\le i \le n_2$. 

For any $X$-stabilizer of $S$, namely $O_1\otimes O_2$ with $O_1=O_{11}\otimes O_{12}$, where $O_{11}$ is restricted to the first $n_2$ locations of the data qubits of $Q_1$, we have 
\[
\overline{\CNOT} (O_1\otimes O_2) \overline{\CNOT}=O_1\otimes (O_{11}O_2).
\]
By the argument (1) above $O_{11}O_2$ is an $X$-stabilizer of $Q_2$. This implies that $O_1\otimes (O_{11}O_2)$ is a valid $X$-stabilizer of $S$.    
Next, for any $Z$-stabilizer of $S$ in the form $O_1\otimes O_2$ with $O_1=O_{11}\otimes O_{12}$, we have 
\[
\overline{\CNOT} (O_1\otimes O_2) \overline{\CNOT}=(O_{11}O_2\otimes O_{12} ) \otimes (O_2).
\]
By argument (2) above, $O_{11}O_2\otimes O_{12}$ is a $Z$-stabilizer of $Q_1$. 
Hence the conjugation by $\overline{\CNOT}$ preserves the set of all $X$- and $Z$-stabilizers of $S$, which implies that  $\overline{\CNOT}$ is a logical operator of $S$. 
Using the property (3) above, the code $S$ has 
logical operators 
\begin{itemize}
    \item $X_LI_L=({\bf X}\otimes  I) \otimes I$ (${\bf X}$ is the transversal $X$ operator on the first $n_1$ qubits of $Q_1$),
    \item $Z_LI_L=({\bf Z}\otimes I) \otimes I$,
    \item $I_LX_L=(I\otimes  I) \otimes {\bf X}$ (${\bf X}$ is the transversal $X$ operator on the $n_2$ qubits of $Q_2$), 
    \item $I_LZ_L=(I\otimes  I) \otimes {\bf Z}$.
\end{itemize}
The conjugation of logical operators of $S$ under the action of $\overline{\CNOT}$ gives
\begin{itemize}
    \item $\overline{\CNOT} (X_LI_L) \overline{\CNOT}=X_LX_L$,
    \item $\overline{\CNOT} (I_LX_L) \overline{\CNOT}=I_LX_L$,
    \item $\overline{\CNOT} (Z_LI_L) \overline{\CNOT}=Z_LI_L$,
    \item $\overline{\CNOT} (I_LZ_L) \overline{\CNOT}=Z_LZ_L$.
\end{itemize}
This implies that $\overline{\CNOT}$ acts as a  logical CNOT operator on $S$. 

(2) At the beginning of the circuit, the direct sum of $Q_1$ and $Q_2$ stabilizes two logical qubits that can be associated to the logical qubits of $Q_1$ and $Q_2$. 
The logical state of this system at the beginning is $\ket{+}_{Q_1} \ket{0}_{Q_2}$.
The logical state right before the middle CNOT is
\[
(\ket{0}_{Q_1}+e^{\frac{i\pi}{\theta}}\ket{1}_{Q_1})\ket{0}_{Q_2},
\]
where the first state is the logical magic state $ \ket{R_Z(\theta)}_{Q_1}$.
The joint state after applying the transversal CNOT is
\begin{equation}\label{eq:bell-like}
\ket{0}_{Q_1}\ket{0}_{Q_2}+e^{\frac{i\pi}{\theta}}\ket{1}_{Q_1}\ket{1}_{Q_2}
\end{equation}
Changing the $\{\ket{0},\ket{1}\}$ basis of $Q_1$ to $\{\ket {+} ,\ket {-}\}$ 
exchanges \eqref{eq:bell-like} to
\begin{equation}
\frac{
\ket{{+}}_{Q_1}(\ket{0}_{Q_2}+e^{\frac{i\pi}{\theta}}\ket{1}_{Q_2})
\;+\;
\ket{{-}}_{Q_1}(\ket{0}_{Q_2}-e^{\frac{i\pi}{\theta}}\ket{1}_{Q_2})
}{\sqrt{2}}.
\label{eq:x-decomp}
\end{equation}
Applying a logical measurement on $Q_1$ in the $X$-basis, presented in the circuit by $M_{X_L}$, projects the state of the second logical qubit of \eqref{eq:x-decomp} into one of:
\[
\ket{0}_{Q_2}+e^{\frac{i\pi}{\theta}}\ket{1}_{Q_2}=\ket{R_Z(\theta)}_{Q_2} \ (\text{Measurement +1}),
\] where no correction is needed. 
The other potential outcome state is 
\[ 
\ket{0}_{Q_2}-e^{\frac{i\pi}{\theta}}\ket{1}_{Q_2}=Z_L\ket{R_Z(\theta)}_{Q_2}\ (\text{Measurement -1}),
\] where a correction is needed. Applying $Z_L$, which is the  logical $Z$ operator of $Q_2$, will correct the state into
\[
\ket{0}_{Q_2}+e^{\frac{i\pi}{\theta}}\ket{1}_{Q_2}=\ket{R_Z(\theta)}_{Q_2}. 
\]
$\hfill \square$    
\end{proof}
It should be stressed that our transversal code switching protocol is valid for a long range of code families including surface codes of an arbitrary minimum 
distance.  
In particular, the smallest instances for surface codes are based on switching between
\begin{itemize}
    \item $[\![15,1,3]\!]$ DSPS code and $[\![9,1,3]\!]$ code for realizing logical $S$ gate, and
    \item $[\![31,1,3]\!]$ triorthogonal code discussed in the previous section and the $[\![9,1,3]\!]$ rotated surface code  for realizing $T$ gate. 
\end{itemize}
One can see Appendix \ref{A:15-1-3 stabilizers} for more information about the structure of these two codes.

Furthermore, one can extend the application to any other $Z$-rotation magic state $\ket{R_Z(\theta)}$. 
For instance, the conditions of the code switching protocol are satisfied by   
\begin{itemize}
    \item the codes generated in Theorem \ref{T:general family form}, including the numerical examples of Table \ref{tab:parameters small}, to implement any $Z$-rotation gate for color codes,
    \item the codes generated in Corollary \ref{cor: tri-surface}, including the numerical examples of Table \ref{tab: DSPS parameters small}, to implement
    any $Z$-rotation gate for rotated surface codes.
\end{itemize}
This extends the set of available logical gates for certain families of quantum codes, such as color codes and surface codes, from Clifford+$T$ to Clifford+$R_Z(\theta)$ gate set.
Such an alternative gate set has a potential resource optimization application in FT implementation of quantum algorithms that rely on multi-controlled gates, including Grover’s search \cite{grover1996fast} and Shor’s factoring \cite{shor1994algorithms} algorithm, as well as quantum chemistry simulations, such as those for the FeMoCo molecule \cite{reiher2017elucidating}.

The proposed code switching scheme can be further improved by analyzing the propagation of errors through the transversal CNOT gates and by updating the logical operator based on the error-location information obtained from some syndrome extraction cycles. 
Such a correction procedure is analogous to the scheme described in \cite[Section~2.2.1]{heussen2025efficient}, and we postpone such details for a future research.

Note also that in \cite{jiao2025low}, a transversal code switching protocol is proposed that is different from what is described above. 
For instance, it relies on the existence of an initial triorthogonal code, and the code switching happens between two codes of the same length. 

%%%%%%%%%%%%%%%%%%%%%%%%%%%%%%%
\subsection{Numerical simulation}
%%%%%%%%%%%%%%%%%%%%%%%%%%%%%%%
This section evaluates the application of the quantum codes introduced in the previous section for preparing non-Clifford magic states. 
Here, we perform an initial benchmarking simulation to illustrate their core properties, and we postpone a comprehensive simulation to future work.

The first step to evaluate the performance of our code-switching protocol, for realizing the non-Clifford magic states, is to evaluate the state preparation performance of our newly designed quantum codes. 
In particular, we are interested in the smallest distance, $d=3$, quantum codes inside of our families of quantum codes with a logical $T$ gate and the local geometry of rotated surface code. 
This includes the  
triorthogonal codes $[\![31,1,3]\!]$ and $[\![39,1,3]\!]$. 

First, we check the operational resilience of the logical state preparation circuits in the mentioned codes. For this purpose, we conduct circuit-level Monte Carlo simulations utilizing the \texttt{qLDPC} framework \cite{perlin2023qldpc} backed by the \texttt{Stim} and \texttt{Sinter} packages \cite{gidney2021stim}. 
Our numerical pipeline assesses the intrinsic vulnerability of each designed encoding circuit $U_{\text{prep}}$ by subjecting it to a circuit-level depolarizing noise model, followed by a noiseless diagnostic verification sequence.

\textbf{Noise model.}
Let $C$ denote a chosen quantum code for the experiment. 
A logical state preparation circuit $U_{\text{prep}}$ for $C$ is a circuit that consists of state initializations, a structured sequence of single- and two-qubit Clifford operations, and auxiliary measurements designed to prepare a targeted logical state $\ket{\psi_L} \in C$. 

We model hardware imperfections via a circuit-level depolarizing noise model parameterized by a uniform physical error rate $p$. The noise channels are injected discretely into the execution timeline as follows:
\begin{enumerate}
    \item \textbf{State initialization:} Every data qubit initialized in the $\ket{0}$ state is ssubjected to a bit-flip noise with probability $p$.
    \item \textbf{Single-qubit and two-qubit gates:} Following each single-qubit (respectively two-qubit) gate, a symmetric single-qubit (respectively two-qubit) depolarizing channel is applied to the target qubit (respectively the control and target qubits).
    \item \textbf{Measurement assignment:} Each measurement is followed by a classical bit-flip error with probability $p$.
\end{enumerate}

To calculate the logical fidelity of the state prepared by the noisy execution of $U_{\text{prep}}$, the software framework appends a mathematically ideal, completely noiseless diagnostic tail $\mathcal{M}_{\text{diag}}$ before parsing the statistical registers. 
In particular, the simulation blocks per Monte Carlo shot can be described as a noisy execution of the circuit $U_{\text{prep}}$ followed by
a noiseless execution directive $\mathcal{M}_{\text{diag}}$ containing:
\begin{itemize}
    \item so called ``flag detectors'' will be added to mid-circuit flag ancilla measurements that were left unaddressed in the circuit.
    \item ideal stabilizer and logical Pauli observable measurements (syndrome extractions) required to stabilize the desired logical state. 
\end{itemize}

\textbf{Data processing and error rate determination.}
The sampler generates a batch of $N_{\text{total}}$ shots containing binary vectors for all detectors and logical observables. 
Since we are incorporating fault-tolerant verification using flag qubits, a post-selection step is applied.
In particular, we consider a trial acceptable if and only if all flag detectors are zero. 
Trials exhibiting non-trivial flag signatures, not all zeros, will be discarded. 
The \textit{Discard Rate} is evaluated as:
\[
    R_{\text{discard}} = \frac{N_{\text{discarded}}}{N_{\text{total}}}
\]
For each accepted shot, the extracted (noiseless) stabilizer syndromes are forwarded to a circuit-level decoder for a follow-up error correction circuit.
Then, the decoder predicts a logical correction operator. 
A \textit{Logical Error Event} is recorded if the decoder's prediction fails, i.e., after performing the correction we end up with a wrong logical state. 
The conditional \textit{Logical Error Rate} is computed strictly over the surviving population as:
\[
    P_{\text{log}} = 1-\frac{N_{\text{success}}}{N_{\text{accepted}}} 
\]
where $N_{\text{accepted}}$ shows the number of shots with a trivial flag signature, and $N_{\text{success}}$ represents the number of shots with trivial flag signatures and correct decoder prediction. 
Finally, simulations at each physical error rate $p$ continue dynamically until reaching either a maximum sample roof of $10^7$ shots or an accumulated threshold of $10^5$ logical failure events.

\subsection{Code-switching between $[\![31,1,3]\!]$ and distance three surface code}

In this section, we discuss the steps required to implement the code-switching protocol described in Theorem~\ref{T: Magic teleportation}, and demonstrate its application to distance $3$ codes, namely the triorthogonal $[\![31,1,3]\!]$ code and the $[\![9,1,3]\!]$ surface code.
In particular, we perform end-to-end circuit-level simulations for fault tolerant magic state preparation (after replacing $T\ket{+}$ with $S\ket{+}$), demonstrating logical error suppression in the low-noise regime using only 45 physical qubits. 
The details of our main steps are reflected below.

\textbf{$\ket{+}$ state preparation.} 
The first step to prepare a logical $\ket{T}$ magic state in a distance three surface code, using the code-switching circuit of Theorem \ref{T: Magic teleportation}, is to prepare a logical $\ket{+}$ 
in a triorthogonal code with a local geometry of rotated surface code. 
Here we take advantage of the $[\![31,1,3]\!]$ and $[\![39,1,3]\!]$ triorthogonal codes that we constructed in Section \ref{Sec: Surface as another}. 
To evaluate the fidelity of the logical $\ket{+}$ state preparation, we first extract a $\ket{+}$ state preparation circuit for each code, and then we add ancilla qubits to make the circuit FT. 
We take advantage of the \texttt{MQT QECC} package \cite{mqt} to design such a circuit. 
In particular, we add three and six mid-circuit flag qubits to the encoding circuits of $[\![31,1,3]\!]$ and $[\![39,1,3]\!]$, respectively to detect certain logical errors. 

Since both of the mentioned triorthogonal codes have distance three, they are single error correcting codes and therefore we use a look-up table decoder.
The logical error rate of this state preparation step can be seen in Figure \ref{fig:plus_prep}. 

\begin{figure}[ht]
    \centering
\includegraphics[width=1\linewidth]{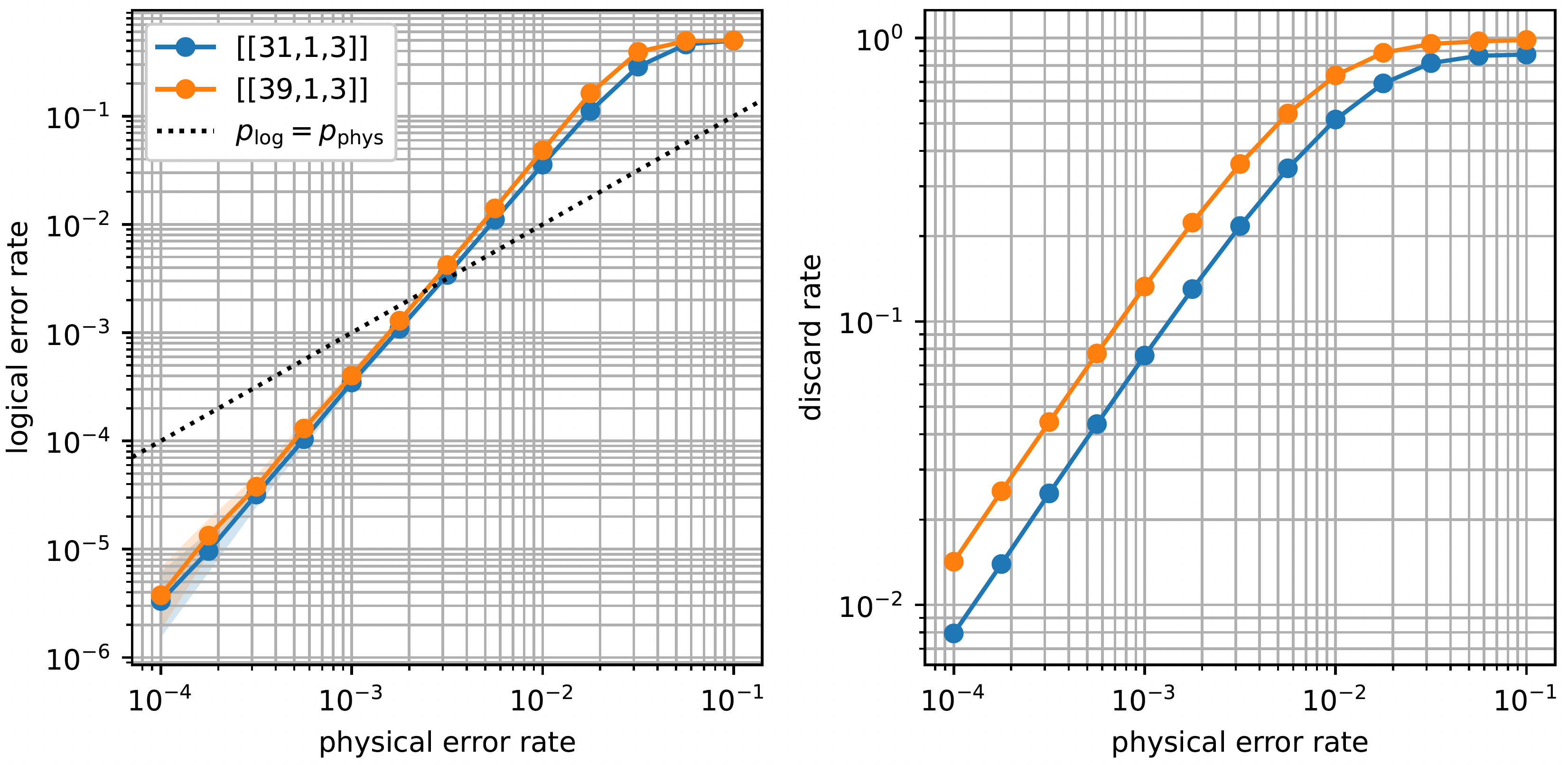}
    \caption{Logical $\ket{+}$ preparation in the codes $[\![31,1,3]\!]$ and $[\![39,1,3]\!]$ with the aid of three and six flag qubits.}
    \label{fig:plus_prep}
\end{figure}
We note that this work does not present an exhaustive search for depth-optimal encoding circuits or generalized FT strategies for such state preparation tasks. 
Discovering alternative encoding circuits with reduced depth and tailored FT protocols offers a clear path to boosting the state preparation performance of these codes. 

\textbf{$\ket{0}$ state preparation.} The second ingredient for executing the code-switching protocol of Theorem \ref{T: Magic teleportation} is to prepare a logical $\ket{0}$ state in the $d=3$ rotated surface code.  
Here again we use a preparation circuit equipped with an extra flag qubit in order to detect and discard certain logical errors. 
The logical error rate of our $\ket{0}$ state preparation circuit is given in Figure \ref{fig:plus_prep2}. 

\begin{figure}[ht]
    \centering
\includegraphics[width=1\linewidth]{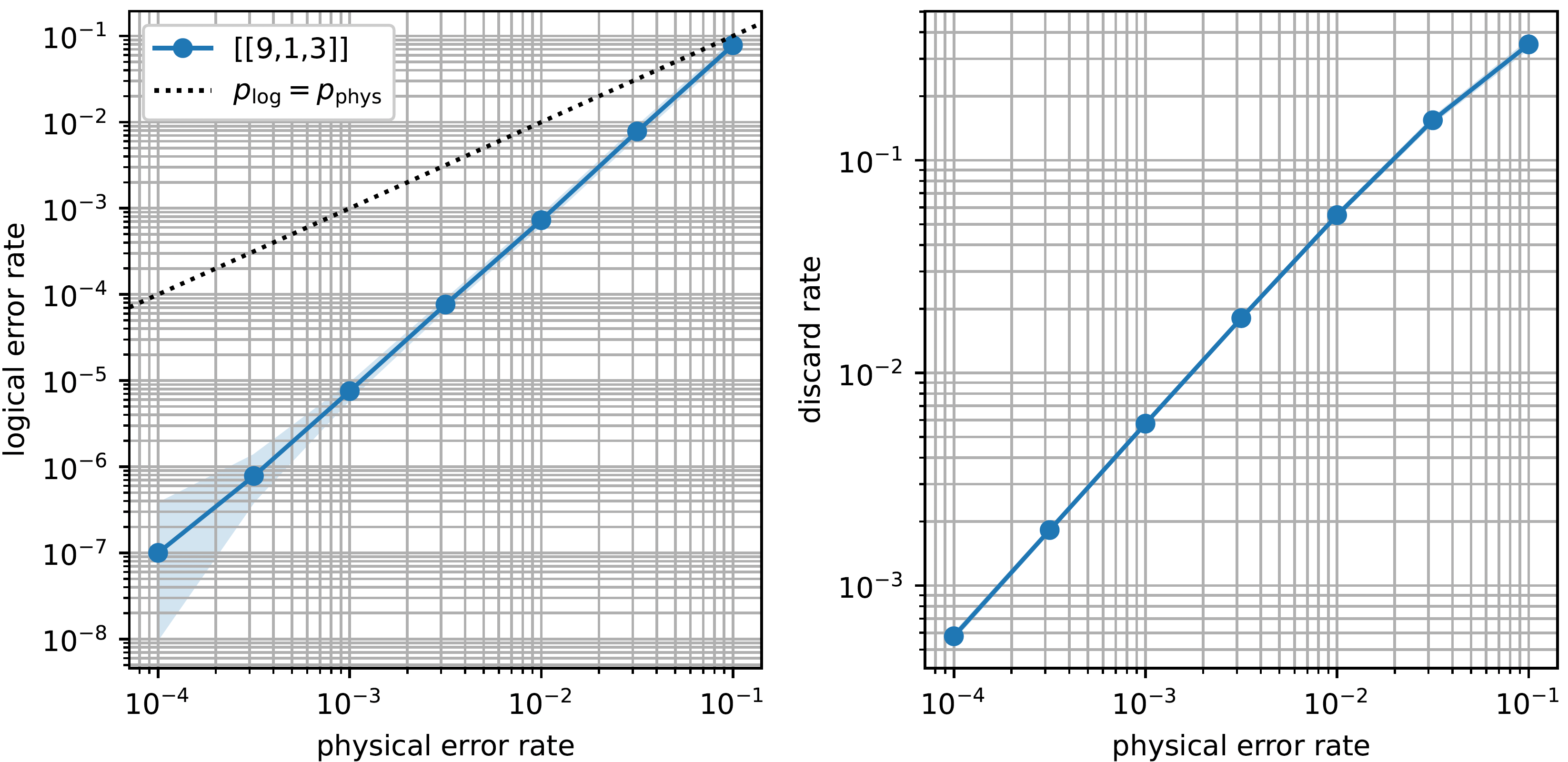}
    \caption{Logical $\ket{0}$ preparation in the $d=3$ rotated surface code with the aid of an ancilla flag qubit.}
    \label{fig:plus_prep2}
\end{figure}

\textbf{Magic state preparation.}
Next, we simulate the end-to-end execution of the code-switching protocol described in Theorem~\ref{T: Magic teleportation} after replacing $T\ket{+}$ with its Clifford proxy $S\ket{+}$.

We first define a $[\![40,2,3]\!]$ quantum code as the direct sum of the $[\![31,1,3]\!]$ code and the $[\![9,1,3]\!]$ rotated surface code. 
We then concatenate the logical $\ket{+}$ and $\ket{0}$ state preparation circuits for the $[\![31,1,3]\!]$ and $[\![9,1,3]\!]$ codes, respectively, to obtain a logical $\ket{+}\ket{0}$ state preparation circuit for the $[\![40,2,3]\!]$ code. 
This circuit consists of 40 data qubits together with four ancilla qubits: three for the FT $\ket{+}$ state preparation of the $[\![31,1,3]\!]$ code and one for the FT $\ket{0}$ state preparation of the $[\![9,1,3]\!]$ code. 
The complete circuit is depicted in Figure \ref{fig:magic_preparation circuit}.

Next, we apply a logical $S$ gate to the first logical qubit, followed by a transversal CNOT between the two logical qubits. 
Finally, we record the outcome of the logical $X$ measurement on the first logical qubit in an additional flag qubit, thereby enforcing the logical $X$ measurement to have a trivial syndrome.

The resulting logical error rate and discard rate of our $S\ket{+}$ state preparation are shown in Figure \ref{fig:plus_prep3}.
Again here we use the look-up table to determine whether the decoder has predicated a correct logical correction, or if a logical error event has happened. 
\begin{figure}[h]
    \centering
\includegraphics[width=1\linewidth]{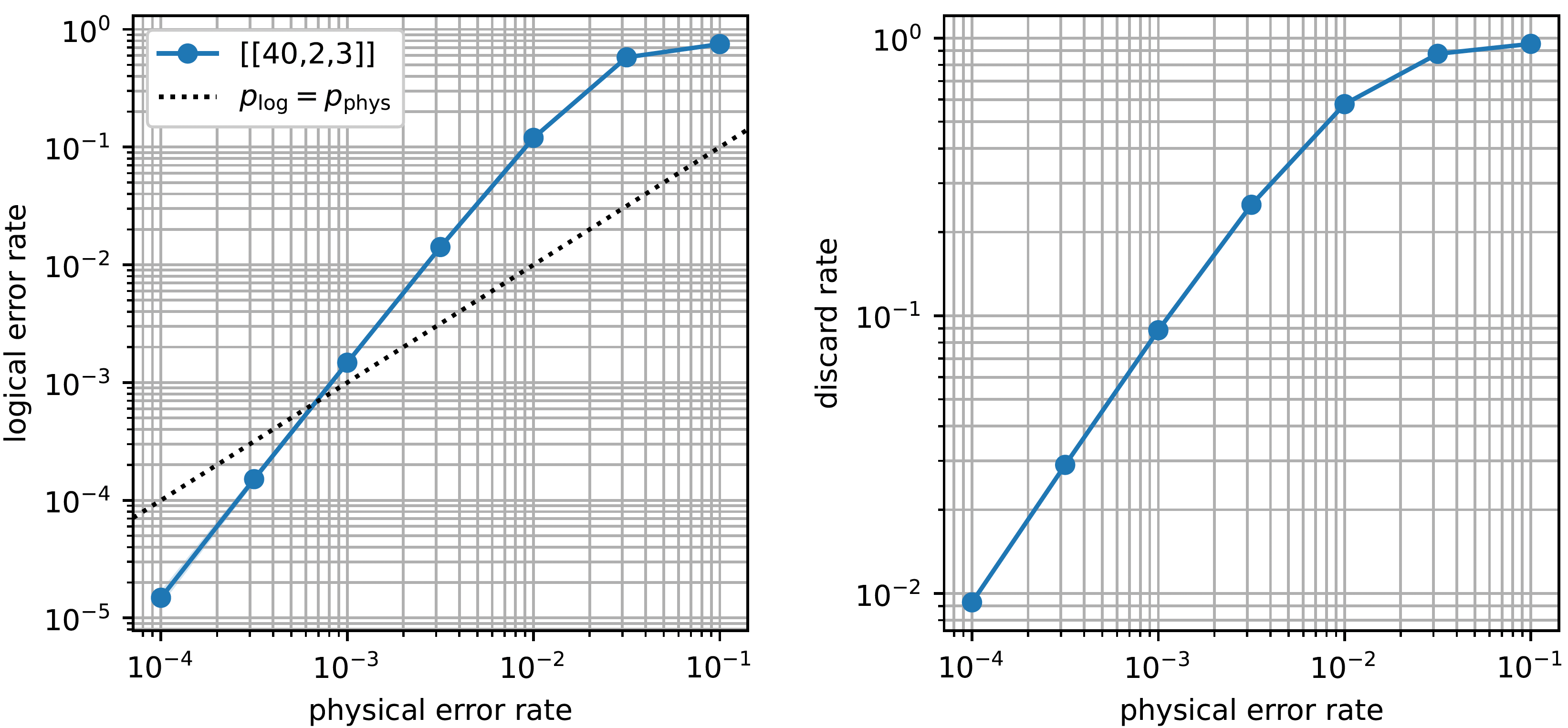}
    \caption{Logical $S\ket{+}$ state preparation as a Clifford proxy for $T\ket{+}$ state preparation in the $d=3$ rotated surface code using the code-switching protocol described in Theorem~\ref{T: Magic teleportation}. 
    The simulation uses a total of 45 qubits, including 5 ancilla qubits.}
    \label{fig:plus_prep3}
\end{figure}
The simulation shows a successful state preparation for physical error rate of below roughly $6\times 10^{-4}$.
At the physical error rate of $10^{-4}$, the protocol successfully prepares a verified logical $S\ket{+}$ state in approximately $99\%$ of attempts, with a conditional logical error rate of only  $\approx 10^{-5}$.

Overall, our simulation uses only 45 qubits,
a resource overhead drastically smaller than traditional surface code distillation factories requiring hundreds of qubits for an equivalent code distance. 

Lastly, although the protocol is demonstrated only for distance-three codes, the simulation validates the complete fault-tolerant implementation of the proposed code-switching scheme. Extending the simulation to larger code distances and evaluating the fidelity of the prepared magic states are natural directions for future work.

\section{Conclusion}

In this work, we developed a unified framework for constructing families of $r$-orthogonal and divisible quantum codes supporting transversal logical $Z$-rotation gates at arbitrary levels of the Clifford hierarchy. 
Beyond developing color-code constructions, our framework produces $r$-orthogonal codes compatible with the local geometry of rotated surface codes, therefore extending the applicability of fault-tolerant code switching well beyond previously known settings. 
Together with an overhead optimization procedure, these constructions provide practical candidates for implementing logical non-Clifford gates within architectures based on rotated surface codes.

Building on these constructions, we generalized the transversal CNOT 
code-switching protocol to this broader class of codes, enabling the fault-tolerant realization of logical $Z$-rotation gates while preserving the favorable hardware properties of rotated surface codes. 
As a proof of principle, we presented the first end-to-end circuit-level demonstration of fault-tolerant magic-state preparation in a distance three rotated surface code using only 45 physical qubits, validating the complete implementation of the proposed protocol.

Several interesting directions remain for future work. 
These include extending the simulations to larger code distances, investigating the resource overhead and error-suppression properties of the protocol, investigating more efficient decoding strategies tailored to the new code families, and characterizing the fidelity of the prepared logical resource states. 
More broadly, the framework introduced here opens the possibility of constructing additional families of locally constrained quantum codes with transversal higher-level Clifford-hierarchy gates, providing new avenues toward scalable fault-tolerant quantum computation.

\section*{Acknowledgment}
We thank other members of the Quantum Information Lab at Tecnun and the Hitachi Cambridge Laboratory for their support. This work was supported by the Spanish Ministry of Economy and Competitiveness through the MADDIE project (Grant No. PID2022-137099NBC44) and by the Basque Quantum (BasQ) strategy through the Decoders for quantum low-density parity-check codes (qLDPCDec) project. 
RD gratefully acknowledges the kind hospitality and fruitful discussions with Gilles Zémor, Elena Bardini, Eduardo Camps Moreno, and Virgile Guemard at the Institut de Mathématiques de Bordeaux, and David Feder, Mohsen Mehrani, and Ali Moradi at the Department of Physics and Astronomy at the University of Calgary. Parts of this research were developed during his visits to these institutions.

%%%%%%%%%%%%%%%%%%%%%%%%%%%%%%%%%%%%%%%
\bibliographystyle{unsrt}
\bibliography{Doubling.references}
%%%%%%%%%%%%%%%%%%%%%%%%%%%%%%%%%%%%%%%

%%%%%%%%%%%%%%%%%%%%%%%%%%%%%%%
\appendix
%%%%%%%%%%%%%%%%%%%%%%%%%%%%%%

\section{Color codes and logical gates}\label{Sec: colandlogics}
%%%%%%%%%%%%%%%%%%%%%%%

Here, we provide a more detailed definition of topological color codes using the language of cell complexes and the materials of \cite{kubica2015universal,Asym5,pesah2026quantum}.
Let $\mathcal{L}$ be a finite, $d$-dimensional cell complex, and let $G$ be its $1$-skeleton (the graph consisting of the vertices and edges of $\mathcal{L}$). We can define a $d$-dimensional color code on $\mathcal{L}$ if the following two conditions are met:
\begin{itemize}
    \item \textbf{Regularity}: The graph $G$ is $(d+1)$-regular, i.e., every vertex is incident to exactly $d+1$ edges.
\item \textbf{Colorability}: The $d$-cells (the highest-dimensional faces) of $\mathcal{L}$ admit a $(d+1)$-coloring such that any two $d$-cells sharing a $(d-1)$-dimensional face are assigned different colors.
\end{itemize}
Given such a complex, the quantum code is constructed as follows.  
A physical qubit is associated with each vertex (0-cell) of $\mathcal{L}$.
The stabilizer group is generated by $X$-type and $Z$-type operators. 
An $X$-type stabilizer is associated with each $x$-cell, and a $Z$-type operator is associated with each $z$-cell, acting non-trivially as Pauli $X$ or $Z$ on the qubits within those respective cells. 
The dimensions of these cells must satisfy 
\[
2 \le x + z \le d + 2.\]
The parameter $d$ defines the topological dimension of the color code. 

To encode logical information, we consider color codes defined on manifolds with boundaries. 
Following the convention of punctured closed lattices, these are constructed by selecting a vertex $v \in \mathcal{L}$ and removing it along with all incident $k$-cells for $1 \le k \le d$. 
Such puncturing creates a single connected boundary, yielding a code that encodes $k=1$ logical qubit.

\subsection{Transversal gates}
For the rest of this section, we discuss the conditions for realizing a logical phase gate through the transversal application of different phase gates on data qubits for arbitrary CSS codes. 
Here the condition is for realizing a logical gate without requiring a correction operator.
For more details and the proof, one can see \cite{webster2023transversal} and \cite{leitch2025transversal}.

\begin{theorem}\label{T:logical con}
Let $Q$ be a quantum CSS code of length $n$ with the $X$-stabilizer generator set $H_X$ and $X$-logical operator $L_X=\{u\}$. Then applying the transversal operator
\[
\bigotimes_{i=1}^n R_Z(\frac{\pi}{2^{r-1}})^{t_i}
\]
to the data qubits realizes a logical $R_Z(\frac{\pi}{2^{r-1}})^{t}$ gate if
\begin{itemize}
    \item $\sum_{i=1}^{n}t_iu_i\equiv t \pmod{2^r}$
    \item $\sum_{i=1}^{n}t_iv_i\equiv 0 \pmod{2^r}$ for each $v \in H_X$, and
    \item Schur product of $(t_1,t_2,\ldots,t_n)$ and any  $2\le s \le r$ number of vectors in $H_X \cup L_X$ has weight divisible by $2^{r-s+1}$.
\end{itemize}
\end{theorem}

%%%%%%%%%%%%%%%%%%%%%%%%%%%%%%
\section{Supplementary materials of Section \ref{Sec: general family}}\label{A:closed formula}
%%%%%%%%%%%%%%%%%%%%%%%%%%%%%%%%%%%

\begin{proof}[Proof of Theorem \ref{T:Doubling-general}]
Following similar steps as in other relevant proofs, for example see \cite{bravyi2015doubled,jain2025transversal,berardini2025asymptotically}, one gets a new quantum code with parameters 
$[\![2n_1+n_2,1,\min \{d_1,d_2+2\}]\!]$.
Moreover, the resulted quantum code corresponds to the CSS code of $C_2 \subseteq C_1$, where $C_2$ satisfies $C_2\subseteq C_1 \subseteq C_2^\bot$. 
This implies that $C_1 \setminus C_2 \subseteq C_2^\bot \setminus C_1^\bot$. Thus, $d_x\ge d_z$ for such code. 
That is why we ignored the $X$-distance in the theorem.

The new quantum code $Q$ has the $X$-stabilizer generators 
\begin{equation*}
 G=\begin{bmatrix}
        E_1 & E_1&\bf{0_{n_2}}\\
        \bf{0_{n_1}}&\bf{0_{n_1}}&E_2\\
        \bf{0_{n_1}}&\bf{1_{n_1}}&\bf{1_{n_2}}
    \end{bmatrix},
\end{equation*}
Note that by Theorem \ref{T:NandS for logical} the code $Q_2$ has a transversal logical $R_Z(\frac{\pi}{2^{r-1}})$ in the form of 
\[
O_2=\bigotimes_{i=1}^{n_2} R_Z(\frac{\pi}{2^{r-1}})^{t},
\]
where $n_2t\equiv 1 \pmod{2^r}$. 
Since $n_1$ is odd, we can choose an operator
\[
O_1=\bigotimes_{i=1}^{n_1} R_Z(\frac{\pi}{2^{r-1}})^{s},
\]
satisfying $n_1s_1\equiv -1 \pmod{2^r}$, for $s_i\in \Z_{2^r}$.
Now the operator 
\[
O=O_1^{-1}\otimes O_1\otimes O_2
\]
gives the desired logical action. 
This is because one can choose the logical $X$-operator to be the all $1$ vector $u$. Now the exponent vector $t$ has the form $(-s_{\bf n_1},s_{\bf n_1},t_{\bf n_2})$. One can check easily that all the Schur products satisfy the conditions of Theorem \ref{T:logical con}. 
For instance,
$u\ast t \equiv 1 \pmod{2^r}$, and the first two blocks generate $2^r$ divisible codes, the third row satisfies 
\[
(0_{n_1},1_{n_1},1_{n_2})\ast t \equiv 1-1=0 \pmod{2^r},
\]
and the rest of Schur products will be implied similarly.
Thus $O$ gives a logical $R_Z(\frac{\pi}{2^{r-1}})$. 
$\hfill \square$
\end{proof}
We just remark that, in the above proof, if $n_1+n_2 \equiv 0 \pmod{2^k}$, one can choose the physical operator to be 
\[
O=\bigotimes_{i=1}^{2n_1+n_2} M,
\]
where $M=(R_Z(\frac{\pi}{2^{k-1}}))^{n_1^{-1}}$.

\begin{proof}[Proof of Theorem \ref{T:Doubling color code1}]
The proof follows from applying the result of Corollary \ref{T:Doubling-selforthogonal} recursive using all-even codes of different lengths and the previously generated quantum code of smaller distance. 
More precisely, to build a quantum code with minimum distance $d=2k-1$, one needs three ingredients, namely: 
\begin{enumerate}
    \item the quantum code of distance $d=2k-3$ obtained from the previous step,
    \item two copies of all-even  code of distance $2k-1$, and
    \item a new check qubit connected to all data qubits of one pair of all-even code, and connected to all data qubits in the top layer of the previous block.
\end{enumerate}
As an example, one can see the structure of $[\![17,1,5]\!]$ in Figure \ref{fig:recursive}. 

We proceed with induction on $k$. The initial cases for $k=2$ and $3$ have already been discussed ($[\![7,1,3]\!]$ and $[\![17,1,5]\!]$ codes). 
So suppose that a 2D quantum self-orthogonal code $Q$ with parameters $[\![2k^2-1,1,2k-1]\!]$ containing $k^2-k-3$ $X$-stabilizer generators of weight four, and the rest of $X$-stabilizer generators having weights
$8,12,16,\ldots,4k-4$. We assume this code has the geometry discussed in the theorem (2D, three colorable, top layer representing a logical $X$ and $Z$).  

We modify the code $Q$ and two copies of the all-even code of length $2k+1$ using the result of Corollary \ref{T:Doubling-selforthogonal} to generate a new quantum code $Q'$ with the $X$ stabilizer generators of the form matrix $G$ shown in \eqref{E:doubling matrix2}. 
Since the top of layer of data qubits in $Q$ represent a minimum weight logical operator, the code $Q'$ has a Tanner graph, where 
\begin{itemize}
    \item 
the bottom part is the Tanner graph of $Q$, \item on the top there are 2 layers of $2k+1$ data qubits connected through $2k$ wight four checks, and 
\item in the middle there is a check qubit connecting a layer of $2k+1$ from the top with the logical operator of $Q$ through a check which has weight $2k+1+2k-1=4k$. 
\end{itemize}
Now Corollary \ref{T:Doubling-selforthogonal} implies that the top layer resembles a minimum weight logical operator, and hence the claimed triangular 2D representation is preserved. 
One can color the middle check with a color that is different from the two colors on the top layer of $Q$, and the new $2k$ checks of all-even codes by the remaining two colors. Hence it is three colorable. 
The self-orthogonality part and the claims about parameters and logical action are other direct consequence of Corollary \ref{T:Doubling-selforthogonal}. 
So we only prove the fact about the weight of stabilizer generators.

The code $Q$ has $k^2-k+3$ weight four generators. 
Adding the new $2k$ weight four generators give the total of
\[k^2+k-3=(k+1)^2-(k+1)-3
\]
weight four stabilizer generators for $Q'$.
Finally, the new wight $4k$ (connecting top layer of $Q$ and all data qubits of one pair of all-even code, called middle check above) stabilizer generator described above, gives a set of $(k+1)-2$ stabilizer generators with weights 
\[8,12,16,\ldots, 4k-4,4k.\]
$\hfill \square$
\end{proof}

\begin{proof}[Proof of Theorem \ref{T:CH3 family}]
We proceed recursively by applying the doubling construction of Theorem \ref{T:Doubling-general} to the following pair of codes; (a) distance $2k+1$ quantum code of Theorem \ref{T:Doubling color code1} and, (b) the distance $2k-1$ quantum code obtained from the previous round of recursion. 

The case with distance $d=1$ (the $[\![1,1,1]\!]$ code) follows trivially, and $d=3$ was already mentioned in Example \ref{example: higherd=3} (the $[\![15,1,3]\!]$ Reed-Muller code). 
Suppose that the result is true for $m=2k-1$, i.e. this approach implies the existence of 
a family of quantum codes with parameters 
\[
[\![\frac{2}{3}k(k+1)(2k+1)-(2k+1),1,2k-1]\!]
\]
that realizes logical $T$ gate through transversal action of powers of $T$ on physical qubits, and all $X$-stabilizers have weights divisible by $8$, for each $1\le m \le 2k-1$. 
Next we prove the result for $m=2k+1$. 
Applying the doubling construction of Theorem \ref{T:Doubling-general} to 
(a) the $[\![n_1=2(k+1)^2-1,1,2k+1]\!]$ self-orthogonal code and (b) the previously obtained \[
[\![n_2=\frac{2}{3}k(k+1)(2k+1)-(2k+1),1,2k-1]\!]
\] 
code gives a new code with parameters 
\[
[\![2n_1+n_2,1,2k+1]\!].
\] 
An easy computation shows that
\[
2n_1+n_2=\frac{2}{3}(k+1)(k+2)(2k+3)-(2k+3).
\]
So we only need to check the claim about the weight of $X$-stabilizers. The set of new $X$ stabilizer generators are in the form  
\[
 G=\begin{bmatrix}
        E_1 & E_1&\bf{0_{n_2}}\\
        \bf{0_{n_1}}&\bf{0_{n_1}}&E_2\\
        \bf{0_{n_1}}&\bf{1_{n_1}}&\bf{1_{n_2}}
    \end{bmatrix},
\]
where $E_1$ and $E_2$ are the $X$-stabilizer generators of the (a) and (b) codes above, respectively. 
Since $E_1$ generates a doubly even code, and $E_2$ a triply even code, we conclude that the first two blocks of matrix $G$ above, and any linear combination of them, always have triply even weights. 
A linear combination of the last row with the other blocks implies an $X$-stabilizer of the form 
\[
u=(v,v+{\bf{1_{n_1}}},w+{\bf{1_{n_2}}}),
\]
where $v$ and $w$ are generated by $E_1$ and $E_2$. Moreover, we have 
\[
\begin{split}
\wt(u)&\equiv n_1+n_2=(2(k+1)^2-1)\\&+(\frac{2}{3}k(k+1)(2k+1)-(2k+1))   \\& 
\equiv \frac{4}{3}k(k+1)(k+2) \equiv 0 
\pmod 8. 
\end{split}
\]
The last equality is because $6 \mid k(k+1)(k+2)$.
Now since it is a triply even code, Theorem \ref{T:Doubling-general} implies the existence of transversal logical $T$ gate and completes the proof. $\hfill\square$
\end{proof}

Below, we give the proof of the identities given in \eqref{equ: recurrence length} and \eqref{equ:divisibility}. 

\begin{lemma}\label{T:recurence-color}
Let $r,k\ge 1$ be two integers and
\[
S_r(k)=2\big(\sum_{i=1}^{k} S_{r-1}(i)\big)-1,
\]
where $S_r(1)=1$ for all $r\ge 1$ and
\[S_1(k)=2k-1.\]
Then for each $r\ge 1$ and $k\ge 2$, we have
\[
S_r(k)=\sum_{i=0}^{r} 2^i \binom{i+k-2}{i}.
\]    
\end{lemma}
\begin{proof}
We proceed by induction on $r \ge 1$.
Base case ($r=1$): By definition, $S_1(k)=2k-1$. Using the summation formula, we have 
\[
\begin{split}
\sum_{i=0}^{1} 2^i \binom{i+k-2}{i} &= \binom{k-2}{0} + 2\binom{k-1}{1} \\&= 1 + 2(k-1) = 2k-1.
\end{split}
\]
Hence the base case holds. 
So assume the hypothesis holds for some $r \ge 1$ and for all $k \ge 2$. 
We evaluate $S_{r+1}(k)$ as
\[
S_{r+1}(k) = 2\left(\sum_{i=1}^{k} S_r(i) \right) - 1.\]
Substituting the inductive hypothesis gives
\[
S_{r+1}(k) = 2 \left( \sum_{i=1}^{k} \sum_{j=0}^{r} 2^j \binom{j+i-2}{j} \right) - 1.
\]
Interchanging the order of summation
\[
S_{r+1}(k) = 2 \left( \sum_{j=0}^{r} 2^j \sum_{i=1}^{k} \binom{j+i-2}{j} \right) - 1.
\]
We simplify the inner sum using the hockey-stick identity, $\sum_{i=r}^n \binom{i}{r} = \binom{n+1}{r+1}$.
For our inner sum, we have
\[
\begin{split}
\sum_{i=1}^{k} \binom{j+i-2}{j} &= \binom{j-1}{j} + \binom{j}{j} + \binom{j+1}{j} \\&+ \dots + \binom{j+k-2}{j}.
\end{split}
\]
Since $\binom{j-1}{j} = 0$, we have 
\[
\sum_{i=1}^{k} \binom{j+i-2}{j} = \binom{j+k-1}{j+1}. \]
Substituting this back into $S_{r+1}(k)$ gives 
\[
S_{r+1}(k) = \left( \sum_{j=0}^{r} 2^{j+1} \binom{j+k-1}{j+1} \right) - 1. \]
Choosing $i = j+1$ in the above sum implies
\[S_{r+1}(k) = \left( \sum_{i=1}^{r+1} 2^i \binom{i+k-2}{i} \right) - 1 \]
Finally, observe that when $i=0$, term of the target summation is
\[
2^0 \binom{0+k-2}{0} = 1.
\]
Thus we can therefore replace the main sum by 
 \[S_{r+1}(k) = \sum_{i=0}^{r+1} 2^i \binom{i+k-2}{i}.\]
This completes the induction. 
$\hfill \square$
\end{proof}

Next we proof the divisibility condition of \eqref{equ:divisibility}. \\
\begin{lemma}\label{lem: divisibility rec}
For integers $r \ge 1$ and $k \ge 2$, let
\[
S_r(k) = \sum_{i=0}^{r} 2^i \binom{k+i-2}{i}.
\]
Then the following identity holds
\begin{equation}\label{equ:divisproof}
S_r(k-1) + S_{r-1}(k) = 2^r \binom{k+r-2}{r}.
\end{equation}
\end{lemma}

\begin{proof}
Applying the identity \eqref{equ:first form} gives
\begin{equation}\label{Equ:relation recurrence}
\begin{split}
S_r(k) &- S_r(k-1)= 2 S_{r-1}(k)
\end{split}
\end{equation}
or equivalently one gets
\[
S_r(k-1) = S_r(k) - 2 S_{r-1}(k).
\]
Substituting this into the left-hand side of \eqref{equ:divisproof}, we get
\[
\begin{split}
S_r(k-1) &+ S_{r-1}(k)
= \big[ S_r(k) - 2 S_{r-1}(k) \big] \\&+ S_{r-1}(k) 
= S_r(k) - S_{r-1}(k).
\end{split}
\]
Applying the result of previous lemma implies that
\[
\begin{split}
S_r(k) & - S_{r-1}(k) 
=  \sum_{i=0}^{r} 2^i \binom{k+i-2}{i}
   \\&- \sum_{i=0}^{r-1} 2^i \binom{k+i-2}{i} = 2^r \binom{k+r-2}{r}.
\end{split}
\]
\end{proof}

\begin{proof}[Proof of Theorem \ref{T:general family form}]
The proof for $r=2$ and $3$ and each $k\ge 1$ was provided in Theorems \ref{T:Doubling color code1} and \ref{T:CH3 family}. 
We proceed by applying the doubling construction of Theorem \ref{T:Doubling-general} to the codes with lengths $S_r(k-1)$ and $S_{r-1}(k)$ for each $r$ and $k$ values. 

We prove by induction on $r$. The base cases $r=2,3$ have already discussed. 
So we assume that the result holds for each $r\le i-1$. 
To prove the result for $i=r$ and each $k\ge 1$, we  apply another induction on $k$. 

The base case $k=1$ correspond to the $[\![1,1,1]\!]$ (and $k=2$ was shown by the codes discussed in Example \ref{example: higherd=3}). 
So we assume that the result holds for $k=j-1$, and we need to show that the existence of the code with parameters $[\![S_i(j),1,2j-1]\!]$. 
Existence of such code consequently completes the claims of both of the inductions.

Applying the doubling of the codes with lengths $S_i(j-1)$ and $S_{i-1}(j)$, obtained from the previous steps yields a quantum code with parameters 
\[
[\![S_{r}(k-1)+2S_{r-1}(k),1,2k]\!].
\]
Equation \eqref{equ:first form} implies that $S_{r}(k-1)+2S_{r-1}(k)=S_r(k)$ as desired.  
The claims about weight of stabilizers and the logical gate now follows from Theorem \ref{T:Doubling-general}.
$\hfill\square$
\end{proof}

%%%%%%%%%%%%%%%%%%%%%%%%%%%%%%
\section{Supplementary materials of Section \ref{Sec: Surface as another}}\label{A:supp-surface}
%%%%%%%%%%%%%%%%%%%%%%%%%%%%%%%%%%%

\begin{proof}[Proof of Theorem \ref{T:doubling-surface}]
The proof is based on recursively applying the result of Proposition \ref{P:doubling-surface}. 
First, recall that the codes constructed from this proposition have the $X$-stabilizer generator matrix of the form 
\begin{equation*}
G = \begin{bmatrix}
        E & E&\bf{0_{n}}\\
        \bf{0_{(d+2)^2}}&\bf{0_{(d+2)^2}}&G_2\\
        \bf{0_{(d+2)^2}}&\bf{1_{(d+2)^2}}&\bf{v}
    \end{bmatrix},
\end{equation*}
where $E$ and $G_2$ are the generator matrices of the $X$-stabilizers of a surface code with distance $d+2$, and the code $Q$ is obtained from the previous recursion (with minimum distance $d$),
respectively. 
Also, $\bf{v}$ is a minimum weight logical $X$-operator of $Q$. 
The vector $(\bf{1_{(d+2)^2}},\bf{0_{(d+2)^2}},\bf{0_n})$ is a logical X-operator of the new quantum code. 
One can show that the minimum weight $X$-logical operator has the form $(u,u,v)$, where $u$ is a minimum weight $X$-logical operator of the surface code.   

The case $d=1$ is trivial (the $[\![1,1,1]\!]$ code). 
For $d=3$, applying the mentioned proposition to the $[\![9,1,3]\!]$ surface code and the totally orthogonal code $[\![1,1,1]\!]$ gives a quantum code with parameters $[\![19,1,3]\!]$. 
In this case, the first two row blocks of $G$ generate a subspace with weight divisible by four. 
The last row is $(0_9,1_9,1_1)$, which has weight $10 \equiv 2 \pmod 4$. 
The orthogonality part can be easily checked by the structure of the rows. 
The claim about the transversal logical $S$ follows immediately as
the $X$-stabilizers are orthogonal. 

Next, we assume that the claim is true for distance $i=d-2$ and one can obtain a quantum code $Q_0$ with the claimed properties. 
Applying the construction of Proposition \ref{P:doubling-surface} to 
the code $Q_0$ with parameters 
\[
[\![n=\frac{(d-2)(d-1)d}{3}-1,1,d-2]\!]
\]
and the $[\![d^2,1,d]\!]$ rotated surface code gives a new quantum code with the parameters 
\[
[\![2d^2+\frac{(d-2)(d-1)d}{3}-1,1,d]\!]
\]
and $X$-stabilizers
\begin{equation*}
G = \begin{bmatrix}
        E & E&\bf{0_{n}}\\
        \bf{0_{d^2}}&\bf{0_{d^2}}&G_2\\
        \bf{0_{d^2}}&\bf{1_{d^2}}&\bf{1_n}
    \end{bmatrix}.
\end{equation*}
A simple calculation shows that
\[
2d^2+\frac{(d-2)(d-1)d}{3}-1=\frac{d(d+1)(d+2)}{3}-1.
\]
Moreover, note that the matrix $G_2$ can be decomposed into $[G_{2}',w]^T$, where rows of $G_2'$ generate a subspace with weights divisible by four, and $w$ has weight equal to $2$ modulo $4$. We can represent the matrix $G$ above as 
\begin{equation*}
G = \begin{bmatrix}
        E & E&\bf{0_{n}}\\
        \bf{0_{d^2}}&\bf{0_{d^2}}&G_2'\\
        \bf{0_{d^2}}&\bf{0_{d^2}}&w\\

        \bf{0_{d^2}}&\bf{1_{d^2}}&\bf{1_n}
    \end{bmatrix}.
\end{equation*}
We consider two cases. 
(a) Suppose that $n+d^2 \equiv 0 \pmod 4$. Then the first two blocks of $G$ always have a weight divisible by four. 
Moreover, adding the last row to the vectors generated by the first two blocks gives new vectors of the form 
\[
u=(a,a+{\bf{1_{d^2}}},b+{\bf{1_n}}).
\]
Since $n+d^2$ and $\wt(b)\equiv 0 \pmod 4$, we conclude that 
\[\wt(u)=\wt(a)+d^2-\wt(a)+n-\wt(b) \equiv 0 \pmod 4.
\]
This shows  that, except the third row, the rest of the rows of $G$ generate a space of vectors with weights divisible by four. 

(b) Suppose that $n+d^2 \equiv 2 \pmod 4$. 
Again the first two blocks of $G$ generate vectors of weight divisible by four. 
Adding such vectors to the sum of the last two rows gives vectors of the form 
\[
u=(a,a+{\bf{1_{d^2}}},b+w+{\bf{1_n}}).
\]
Since $n+d^2,\wt(b+w)\equiv 2 \pmod 4$, we conclude that 
\[\wt(u)=d^2+n-\wt(b+w) \equiv 0 \pmod 4.
\]
This completes the proof by partitioning the $X$-stabilizers to a set with weights divisible by four and a set with weights equal to $2$ modulo $4$. 
$\hfill\square$
\end{proof}

For the rest of this section, we deal with the proof of Theorem \ref{T:surface-hierarchy}. 
The process of the proof is very similar to that of Theorem \ref{T:general family form}, and occasionally we refer to the later for more details. 

\begin{proof}[Proof of Theorem \ref{T:surface-hierarchy}]

First recall that applying the result of Theorem \ref{T:Doubling-general} to rotated surface codes led to the family of quantum codes presented in Theorem \ref{T:doubling-surface}. 
We represent length of the codes obtained this way by 
\[
S_2(k)=\frac{(2k-1)(2k)(2k+1)}{3}
\]
i.e the codes inside this family have parameters $[\![S_2(k),1,2k-1]\!]$. 
Furthermore, using the same convention, we can represent $S_1(k)$ to be the length of a surface code with distance $2k-1$, i.e., 
\[\![\![S_1(k)=(2k-1)^2,1,2k-1]\!].  
\]
Repeated application of Theorem \ref{T:Doubling-general} implies a hierarchy of codes, and using induction on first $r$ and then on $k$ one can show that such codes have parameters
$[\![S_r(k),1,2k-1]\!]$, where $r,k \ge 1$ and 
$S_r(k)=2S_{r-1}(k)+S_{r}(k-1)$. 
Note also that we choose $S_r(1)=1$, corresponding to the trivial $[\![1,1,1]\!]$ code. 
Moreover, using the same induction steps one obtain the following matrix for $X$-stabilizers
\begin{equation*}
G = 
\begin{bmatrix}
    E & E & \mathbf{0}_{S_r(k-1)} \\
    \mathbf{0}_{S_{r-1}(k)} & \mathbf{0}_{S_{r-1}(k)} & G_2 \\
    \mathbf{0}_{S_{r-1}(k)} & \mathbf{1}_{S_{r-1}(k)} & \mathbf{1}_{S_r(k-1)}
\end{bmatrix},
\end{equation*}
with a logical operator $(\mathbf{1}_{S_{r-1}(k)},\mathbf{0}_{S_{r-1}(k)}, \mathbf{0}_{S_r(k-1)})$,
where by induction hypothesis $G_2$ is $r$-orthogonal and $E$ is $(r-1)$-orthogonal. Now, one can easy verify that $G$ with the give logical operator forms an $r$-orthogonal quantum code.

The last remaining piece is to prove the claimed formula for $S_r(k)$. 
\begin{lemma}
For each integer $r$ and $k \geq 1$, 
if $S_1(k) = (2k-1)^2$, $S_r(1)=1$, and $S_r(k)=2S_{r-1}(k)+S_{r}(k-1)$, then
\[
S_r(k) = 2^{r+2} \binom{k+r-1}{r+1} + \sum_{i=0}^{r-1} 2^i \binom{k+i-2}{i}.
\]
\end{lemma}
\begin{proof}
We prove by induction on $r$. 
The base case for $r=1$ implies that for any $k$ the formula gives
\[
S_1(k)=8\frac{k(k-1)}{2}+1=(2k-1)^2.
\]
Next we assume that the result is true for all integers equal or smaller than $r-1$, and we prove it for $r$. We proceed with another induction on $k$. When $k=1$, we get $S_r(1)=1$ (because $\binom{0}{-1}=1$) which is true. 
So we assume that the result holds for any integers $i\le r-1$ and any $j\le k-1$, and we need to prove the formula for $S_r(k)$. 
Note that $S_r(k)=2S_{r-1}(k)+S_{r}(k-1)$.
Also by the induction hypothesis we have 
\[
S_{r-1}(k)= 2^{r+1} \binom{k+r-2}{r} + \sum_{i=0}^{r-2} 2^i \binom{k+i-2}{i}
\]
and
\[
S_r(k-1) = 2^{r+2} \binom{k+r-2}{r+1} + \sum_{i=0}^{r-1} 2^i \binom{k+i-3}{i}.
\]
Then we have: 
\[
\begin{split}
S_r(k)&=2S_{r-1}(k)+S_{r}(k-1)=\\& 2^{r+2}\big( \binom{k+r-2}{r} + \binom{k+r-2}{r+1} \big)\\&+   \big(2 \sum_{i=0}^{r-2} 2^i \binom{k+i-2}{i} + \sum_{i=0}^{r-1} 2^i \binom{k+i-3}{i} \big)\\&=
2^{r+2} \binom{k+r-1}{r+1} + \sum_{i=0}^{r-1} 2^i \binom{k+i-2}{i},
\end{split}
\]
where the first sum follows from Pascal's identity and the second sum from \eqref{Equ:relation recurrence} (here the terms in the sum translate to the $2S_{r-1}(k)+S_{r}(k-1)$ of \eqref{Equ:relation recurrence}).
\end{proof}

\end{proof}

%%%%%%%%%%%%%%%%%%%%%%%%%%%%%%%%%%%%%%%%%%%%%%%%%%
\section{Symplectic form and proof of Theorem \ref{T:Surface optimization}} \label{A:symplectic}
%%%%%%%%%%%%%%%%%%%%%%%%%%%%%%%%%%%%%%%%%%%%%%%%%%

We will use the properties of symplectic vector spaces in the proof of Theorem \ref{T:Surface optimization}. Therefore we highlight important properties of symplectic spaces before proving the theorem. 
For more information about symplectic bilinear forms, one can see, for example, \cite[pp.~284--285]{bierbrauer2016introduction}.
\begin{definition}
A {\em symplectic vector space} is a pair $(V,\omega)$ where $V$ is a vector space over a field $\mathbb{F}$ and 
$\omega : V \times V \to \mathbb{F}$ is a bilinear form that is
\begin{itemize}
    \item alternating, i.e., $\omega(v,v)=0$ for all $v \in V$, and
    \item non-degenerate, i.e., if $\omega(v,w)=0$ for all $w \in V$, then $v=0$.
\end{itemize}
The form $\omega$ is called a symplectic form.
\end{definition}
A symplectic vector space 
of dimension $2n$ has a basis
\[
\{e_1,\dots,e_n,f_1,\dots,f_n\}
\] 
such that
\[
\omega(e_i,e_j)=0, \quad 
\omega(f_i,f_j)=0, \quad 
\omega(e_i,f_j)=\delta_{ij}.
\]
Such a basis is called a \emph{symplectic basis}.

In the proof, we view the binary space of all $X$-stabilizers as a binary vector space endowed with the Euclidean inner product, and we show that this structure gives rise to a symplectic form.
\begin{proof}[Proof of Theorem \ref{T:Surface optimization}]
(1) The distance $d$ surface code has parameters $[\![d^2,1,d]\!]$ consisting of $\frac{d^2-1}{2}$ X-checks generators that all have even weight. Let $V_1$ be the space of all X-checks. 
Note that if $v\in V_1$ is orthogonal to all $X$-stabilizers, then it has to be a $Z$-stabilizer. However,  $X$- and $Z$-stabilizers of surface code are disjoint. 
This implies that $V_1$ is a symplectic of dimension $\frac{d^2-1}{2}$.     

Similarly, the code space of all-even vectors of length $\frac{d^2+1}{2}$, namely $V_2$, forms a symplectic form of dimension $\frac{d^2-1}{2}$.    

Let 
\[
\{e_1,\dots,e_r,f_1,\dots,f_r\}
\] 
and 
\[
\{e_1',\dots,e_r',f_1',\dots,f_r'\}
\] 
be symplectic basis of $V_1$ and $V_2$, respectively. 
Then the length $\frac{3d^2+1}{2}$ vectors 
\[
\{(e_1,e_1'),\dots, (e_r,e_r'),(f_1,f_1'),\dots,(f_r,f_r')\}
\]
are pairwise orthogonal and form a self-orthogonal binary space. 

(2) As we mentioned earlier, the claimed quantum code $Q$ has X-check generator
\begin{equation*}
 G=\begin{bmatrix}
        E & E''&\bf{0_{n}}\\
        \bf{0_{d^2}}&\bf{0_{d^2-2d+2}}&G\\
        \bf{0_{d^2}}&\bf{1_{d^2-2d+2}}&\bf{1_{n}}
    \end{bmatrix},
\end{equation*}
where $E''$ is obtained from puncturing the $2d-2$ columns of matrix $E$ corresponding to data qubits that overlap X-checks with weight 4 only once. 
We also fix an $X$-logical operator for $Q$ in the form of 
$(\bf{1_{d^2}},\bf{0_{d^2-2d+2}},\bf{0_{n}})$. Hence the quantum code $Q$ is the CSS code of $C_2$ generated by the rows of $G$ and $C_1$ is generated by the rows of $G$ along with $(\bf{1_{d^2}},\bf{0_{d^2-2d+2}},\bf{0_{n}})$.  

First we prove the self-orthogonality of the matrix $G$, which implies that $Q$ is $X$-self-orthogonal. 
Note that the last row of $G$ is orthogonal to any other rows ($E''$ is an even code). 
Let $(x,y,0)$ and $(x',y',0)$ be two vectors belonging to the first row block. Then both have even weights, and the way $E''$ is constructed, imply that $x\cdot x'=y\cdot y'$. Thus the block $[E \ E'' \ 0]$ forms self-orthogonal vectors.  

Note also that the minimum distance of this CSS code is minimum of $d_x$ and $d_z$, where
\[
d_x=\min(\wt(C1\setminus C_2))  
\]
and 
\[
d_z=\min(\wt(C_2^\bot\setminus C_1^\bot)).
\]
Moreover, since $C_2$ is self-orthogonal, we have 
\[
(C_1 \setminus C_2)\subseteq (C_2^\bot\setminus C_1^\bot).
\]
Hence $d_z\le d_x$. So we only need to show that $d_z=d$. 
If $x$ is a weight $d$  $Z$-logical operator, then $(x,0,0) \in (C_2^\bot\setminus C_1^\bot)$, which show that $d_z\le d$. 
Next let $(x_1,x_2,x_3)\in (C_2^\bot\setminus C_1^\bot)$. This implies that $\wt(x_1)$ is odd. 
Then we have one of the following two cases.
\begin{itemize}
    \item Suppose that $\wt(x_2)$ is even. Then adding $x_2$ to the appropriate locations of $x_1$ gives a $Z$-logical operator of surface code. 
    Thus $d\ge \wt(x_1)+\wt(x_2)$. 
    \item Suppose that $\wt(x_2)$ is odd. Then $\wt(x_3)$ is odd, and thus it is a $Z$-logical operator of $Q_2$. Thus $(x_1,x_2,x_3)$ has weight at least $d$.
\end{itemize}
Hence in any case, the code $Q$ has minimum distance $d$. 
\end{proof}

%%%%%%%%%%%%%%%%%%%%%%%%%%%%%%
\section{Stabilizers of $[\![15,1,3]\!]$ orthogonal and $[\![31,1,3]\!]$ triorthogonal codes}\label{A:15-1-3 stabilizers}
%%%%%%%%%%%%%%%%%%%%%%%%%%%%%%%%%%%
Previously we constructed a quantum orthogonal code by applying a doubling approach to the distance three surface code, the all-even  code of size five, and the totally orthogonal code of $[\![1,1,1]\!]$. 
Below we present the generators of $X$ stabilizers in binary form: 
\[
\begin{split}
& [0, 0, 1, 0, 0, 1, 0, 0, 0, 0, 0, 0, 1, 1, 0]\\
&[0, 0, 0, 1, 0, 0, 1, 0, 0, 1, 1, 0, 0, 0, 0]\\
&[1, 1, 0, 1, 1, 0, 0, 0, 0, 0, 1, 1, 0, 0, 0]\\
&[0, 0, 0, 0, 1, 1, 0, 1, 1, 0, 0, 1, 1, 0, 0]\\
&[0,0, 0, 0, 0, 0, 0, 0, 0, 1, 1, 1, 1, 1, 1].
\end{split}
\]
The first 9 columns span the $X$-stabilizers of $[\![9,1,3]\!]$, the next 5 columns form a generator matrix for all-even  code of size five with a row permutation to match the orthogonality of surface code. 
A minimum weight logical $X$ operator has the form of a minimum weight $X$-logical operator of the rotated surface code expanded by six extra zeros. For instance, the following is an $X$-logical operator of the $[\![15,1,3]\!]$ code 
\[
[1, 0, 0, 1, 0, 0, 1, 0, 0, 0, 0, 0, 0, 0, 0].
\]
This code has $d_X=5$ and $d_Z=3$ as the set of $Z$-stabilizers is larger. Moreover, this code is degenerate to both $X$ and $Z$ distances (which are respectively 4 and 2).

Applying another round of doubling to this code implies the tri-orthogonal code with parameters $[\![31,1,3]\!]$. 
One can realize a {\em logical T gate} for this code by applying the transversal operator that acts as 
\begin{equation}\label{equ:logicalT DSPS}
T_L=\displaystyle\prod_{i=1}^{31}O_i,
\end{equation}
where
\begin{itemize}
    \item $O_i=I$ for $i=1,2,8,9,16,17,24,25$,
    \item $O_i=T^{\dagger}$ for $i=15,30,31$, and
    \item $O_i=T$ for the remaining qubits.
\end{itemize}
Below we provide the $X$-stabilizer matrix of the $[\![31,1,3]\!]$ code.

Since the all ones vector is a logical $X$-operator, one can build such quantum code using this information. 
\vspace{1cm}
\begin{widetext}
\[
\left(
\begin{array}{ccccccccccccccccccccccccccccccc}
0&0&1&0&0&1&0&0&0&0&0&0&1&1&0&0&0&1&0&0&1&0&0&0&0&0&0&1&1&0&0\\
0&0&0&1&0&0&1&0&0&1&1&0&0&0&0&0&0&0&1&0&0&1&0&0&1&1&0&0&0&0&0\\
1&1&0&1&1&0&0&0&0&0&1&1&0&0&0&1&1&0&1&1&0&0&0&0&0&1&1&0&0&0&0\\
0&0&0&0&1&1&0&1&1&0&0&1&1&0&0&0&0&0&0&1&1&0&1&1&0&0&1&1&0&0&0\\
0&0&0&0&0&0&0&0&0&1&1&1&1&1&1&0&0&0&0&0&0&0&0&0&1&1&1&1&1&1&0\\
0&0&0&0&0&0&0&0&0&0&0&0&0&0&0&1&1&1&1&1&1&1&1&1&1&1&1&1&1&1&1\\
\end{array}
\right)
\]

\begin{figure}
\begin{center}
\includegraphics[width=1.1\linewidth]{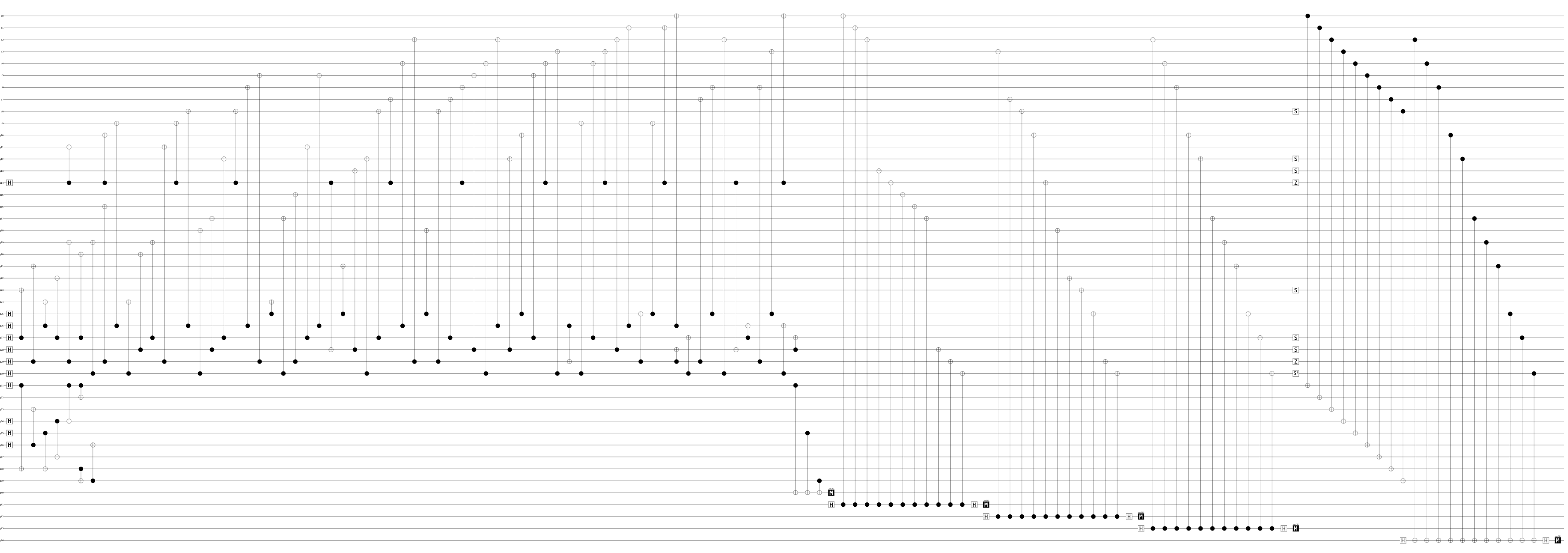}
    \caption{Circuit for $S\ket{+}$ preparation via code-switching between $[\![31,1,3]\!]$ and $d=3$ rotated surface code. The last five qubits are ancilla flag qubits.}
    \label{fig:magic_preparation circuit}
\end{center}
\end{figure}
\end{widetext}

\end{document}